\documentclass[preprint,12pt,authoryear]{elsarticle}

\usepackage{graphicx}

\usepackage{latexsym}
\usepackage[english]{babel}

\usepackage{subfigure}
\usepackage[T1]{fontenc}
\usepackage[utf8]{inputenc}
\usepackage{hyperref}
\usepackage{setspace}
\usepackage{graphicx}
\usepackage{float} 
\usepackage{amssymb}
\usepackage{amsthm}
\usepackage{mathtools}
\usepackage{xcolor}
\usepackage{enumitem}
\usepackage[hmargin=2.5cm,vmargin=2.5cm]{geometry}
\usepackage{natbib}
\theoremstyle{theorem}

\newtheorem{lemma}{Lemma}[section]

\newtheorem{theorem}[lemma]{Theorem}
\newtheorem{prop}[lemma]{Proposition}
\newtheorem{rem}[lemma]{Remark}
\newtheorem{definition}[lemma]{Definition}

\newcommand\Pp{{\mathbb P}}
\newcommand\cP{{\mathcal P}}
\newcommand\cM{{\mathcal M}}

\newcommand\E{\mathbb{E}}
\newcommand\R{\mathbb{R}}
\newcommand\N{\mathbb{N}}
\newcommand\beqnn{\begin{eqnarray}}
\newcommand\eeqnn{\end{eqnarray}}
\newcommand\be{\begin{equation}}
\newcommand\ee{\end{equation}}
\newcommand\eps{\epsilon}

\journal{Stochastic Processes and their Applications}

\begin{document}

\begin{frontmatter}



\title{How does geographical distance translate into genetic distance? }

  
 \author[label1,label2]{Ver\'onica Mir\'o Pina\corref{cor1}}
 \ead{veronica.miro-pina@college-de-france.fr}
  \author[label1,label2]{Emmanuel Schertzer}
 \address[label1]{Laboratoire de Probabilit\'es, Statistique et Mod\'elisation (LPSM), Sorbonne Universit\'e, CNRS UMR 8001, Case courrier 158, 4, place Jussieu 75252 PARIS Cedex 05}
 \address[label2]{Centre Interdisciplinaire de Recherche en Biologie (CIRB), Coll\`ege de France, CNRS UMR 7241, PSL Research University, 11, place Marcellin Berthelot, 75231 PARIS Cedex 05}
\cortext[cor1]{Corresponding author}

\begin{abstract}
Geographic structure can affect patterns of genetic differentiation and speciation rates.
In this article, we investigate  the dynamics of  genetic distances in a geographically structured metapopulation.
We model the metapopulation as a weighted directed graph, with  $d$ vertices corresponding to $d$ subpopulations that evolve according to an individual based model. The dynamics of the genetic distances is then controlled by two types of transitions -mutation and migration events. We show that, under a rare mutation - rare migration regime, intra subpopulation diversity can be neglected and our model can be approximated by a population based model. We show that under a large population - large number of loci limit, the genetic distance between two subpopulations converges to a deterministic quantity that can asymptotically be expressed in terms of the hitting time between two random walks in the metapopulation graph. Our result shows that the genetic distance between two subpopulations does not only depend on the direct migration rates between them but on the whole metapopulation structure.

\end{abstract}

\end{frontmatter}

\section{Introduction}
\label{intro}

\subsection{Genetic distances in structured populations. Speciation} In most species, the geographical range is much larger than the typical dispersal distance of its individuals. 
A species is usually structured into several local subpopulations with limited genetic contact. Because migration only connects neighbouring populations, more often than not, populations can only exchange genes indirectly, by reproducing with one or several intermediary populations. As a consequence, the geographical structure tends to buffer the homogenising effect of migration, and as such, it is considered to be one of the main drivers for the persistence of genetic variability within species (see \cite{malecot} or \cite{ Karlin}).

 The aim of this article is to present some analytical results on the genetic composition of a species emerging from a given geographical structure. The main motivation behind this work is to study speciation.  When two populations accumulate enough genetic differences, they may become reproductively isolated, and therefore considered as different species. As the geographic structure of a species is one of the main drivers for the genetic differentiation between subpopulations, this work should shed light on which are the geographic conditions under which new species can emerge.

 Several authors have studied parapatric speciation, i.e. speciation in the presence of gene flow between subpopulations, for example \citet{gavri1998, Ga2000a, gavri2000} and \citet{Yama,Yama2}.
 In their models, some loci on the chromosome are responsible for reproductive isolation. These loci may be involved in incompatibilities at any level of biological organisation (molecular, physiological, behavioural etc) and either prevent mating (pre-zygotic incompatibilities) or prevent the development of hybrids (post-zygotic incompatibilities).
The number of segregating loci increases through the accumulation of mutations, and decreases after each migration event (creating the opportunity for some gene exchange between the migrants and the host population). 
When the number of segregating loci between two individuals reaches a certain threshold,  they become reproductively incompatible. For example, \citet{Yama, Yama2} studied  the case of a metapopulation containing  two homogeneous subpopulations. The authors studied how the genetic distance, defined as the number of loci differing between the two subpopulations, evolves through time, using a continuous-time model.  When considering metapopulations with more than two subpopulations, this kind of dynamics may translate into complex patterns of speciation. One particularly intriguing example is the case of ring species \citep{noest, gavri1998}, where two neighbouring subpopulations are too different to be able to reproduce with one another but can exchange genes indirectly, by reproducing with a series of intermediate subpopulations that form a geographic `ring'. How these patterns emerge and are maintained is still poorly understood, and we hope that our analytical result might shed some new light on the subject.

\subsection{Population divergence and fitness landscapes}\label{sect:fitness} To study speciation by accumulation of genetic differences, we model the evolution of some loci on the chromosome, that are potentially involved in reproductive incompatibilities. To visualise these evolutionary dynamics, \cite{wright1932} suggested the metaphor of adaptive landscapes. Adaptive landscapes represent individual fitness as a function defined on the genotype space, which is a multi-dimensional space representing all possible genotypes. Wright emphasised the idea of `rugged' adaptive landscapes, with peaks of fitness representing species and valleys representing unfit hybrids. Speciation, seen as a population moving from one peak to another, implies a temporary reduction in fitness, which is not very likely to occur in large populations, where genetic drift is not important enough to counterbalance the effect of selection (see \cite{gavri-holey} for a more detailed discussion). However, \cite{gavri-holey} suggested the idea of `holey' adaptive landscapes, where local fitness maxima can be partitioned into connected sets (called evolutionary ridges). Speciation is therefore seen as a population diffusing across a ridge, by neutral mutation steps, until it stands at the other side of a hole. Theoretical models, such as \cite{gavri1997}, have shown, using percolation theory, that in high-dimensional genotype spaces, fit genotypes are typically connected by evolutionary ridges. 

Our model (see Section \ref{sub-IBD}) is built in this framework. In fact we will assume that, in large populations, deleterious mutations are washed away by selection at the micro-evolutionary timescale and describe the evolutionary dynamics for our set of incompatibility controlling loci as {\it neutral} (any genotype on the evolutionary ridge can be accessed by single mutation neutral steps). This is the idea behind the description of our model in Section \ref{sub-IBD}.

Further, we consider that the evolutionary dynamics along the ridge are slow (as random mutations are very likely to be deleterious, mutations along the evolutionary ridge are assumed to be rarer than in the typical population genetics framework), which is why we study our model in a low mutation - low migration regime (see Section \ref{sect:scaling} for more details). This assumption is commonly made when studying speciation, for example in \cite{Ga2000a} or \cite{Yama}.

\subsection{An individual based model (IBM)} \label{sub-IBD}We model the metapopulation as a weighted directed graph with $d$ vertices, corresponding to the different subpopulations. Each directed edge $(i,j)$ is equipped with a migration rate in each direction. (In particular, if two subpopulations are not connected, we assume that the migration rates are equal to 0.) 
We assume the existence of two scaling parameters, $\gamma$ and $\eps$, that will converge to $0$ successively (first $\gamma \to 0$ and then $\eps \to 0$, see Section \ref{sect:scaling} for more details).

Each subpopulation consists of $n_i^{\eps}$ individuals, $i \in E:= \{1, \dots, d\}$. Each individual carries a single chromosome of length 1, which contains $l^{\eps}$  loci of interest (that are involved in reproductive incompatibilities). We assume that the vector of  positions for those loci -- denoted by ${\cal L}^{\eps} = \{x_{1}, \dots, x_{l^{\eps}}\}$  -- is obtained by throwing $l^\eps$ uniform random variables on $[0,1]$. The positions are chosen randomly at time 0, but are the same for all individuals and do not change through time. Recall that the upper indices (such as in $l^{\eps}$ and ${\cal L}^{\eps}$) are indices and not exponents).
 
Conditioned on ${\cal L}^\eps$, each subpopulation then evolves according to an haploid neutral Moran model with recombination.
\begin{itemize}
\item Each individual $x$ reproduces at constant rate $1$ and chooses a random partner $y$ ($y \ne x$). Upon reproduction, their offspring replaces a randomly chosen individual in the population. 
\item The new individual inherits a chromosome which is a mixture of the parental chromosomes. Both parental chromosomes are cut into fragments in the following way: we assume a Poisson Point Process of intensity $\lambda$ on $[0,1]$. Two loci belong to the same fragment iff there is no atom of the Poisson Point Process between them. For each fragment, the offspring inherits the fragment of one of the two parents chosen randomly.\end{itemize}

 To our Moran model we add two other types of events:
\begin{itemize}
\item \textbf{Mutation} occurs at rate $ b^{\gamma, \eps}$ per individual, per locus according to an infinite  allele model.
\item \textbf{Migration} from subpopulation $i$ to subpopulation $j$ occurs at rate $ m^{\gamma}_{ij}$. At each migration event, one individual migrates from subpopulation $i$ to  $j$, and replaces one  individual chosen uniformly at random in the resident population. (We set $\forall i \in E, \ m^{\gamma}_{ii} = 0$.)

\end{itemize}

We define the genetic distance between two individuals $x$ and $y$ at time $t$ as:
\begin{equation*}
\delta^{\gamma, \eps}_t(x,y) = \frac1{l^\eps} \#\{\ k \in \{1, \dots, l^{\eps}\} : \textrm{ $x$ and $y$ differ at locus $k$ }\}.
\end{equation*}

Consider two subpopulations $i$ and $j$ and let $\{i_1, \dots, i_{n^\eps_i} \}$ be the individuals in population $i$  and $\{j_1, \dots, j_{n^\eps_j} \}$ the individuals in population $j$. The genetic distance between subpopulations $i$ and $j$ at time $t$ is defined as follows:
\begin{eqnarray}\label{def:genetic-distance}
d^{\eps, \gamma}_t(i,j) = \left (\frac1{n_i^{\eps}} \sum_{x \in \{i_1, \dots, i_{n^\eps_i}\}} \min_{y \in \{j_1, \dots, j_{n^\eps_j}\}} \delta_t^{\gamma, \eps}(x,y) \right)  \vee \left ( \frac1{n_j^{\eps}} \sum_{y \in \{j_1, \dots, j_{n^\eps_j}\}} \min_{x \in \{i_1, \dots, i_{n^\eps_i}\}} \delta_t^{\gamma, \eps}(x,y) \right ).
\end{eqnarray}
This corresponds to the so-called modified Haussdorff distance between subpopulations, as introduced by \cite{dubuisson1994modified}. (This distance has the advantage of averaging over the individuals in each subpopulation, so  introducing a single mutant or migrant would produce a smooth variation in the genetic distances.)

\subsection{Slow mutation--migration and large population--dense site regime.}\label{sect:scaling} 
In this section, we start by describing in more details the slow mutation--migration regime alluded to in Sections \ref{sect:fitness} and \ref{sub-IBD}.

It is well known that in the absence of mutation and migration, the neutral Moran model describing the dynamics at the local level reaches fixation in finite time i.e. after a finite amount of time the population becomes homogeneous.
The average time to fixation for a single locus is of the order of the size of the subpopulation \citep{Kimura1, Kimura2} (In our multi-locus model, it will also depend on the number of loci and on the recombination rate $\lambda$.) Heuristically, if we assume a low mutation - low migration regime, i.e. that
\begin{equation}
\forall i, j \in E, \ \ \ \frac1{b^{\gamma, \eps} n^{\eps}_i}, \  \frac1{m_{ij}^{\gamma}} \ \gg \ n^{\eps}_j, \  l^{\eps} \gg 1,
\label{order-parameters} 
\end{equation}
the average time between two migration events ($1/{m_{ij}^{\gamma}} $), and the average time between two successive mutations at a given locus ($1/({b^{\gamma, \eps}n_j^{\eps}})$) are much larger than the average time to fixation.
This ensures that the fixation process is fast compared to the time-scale of mutation and migration, and, as a result, when looking at a randomly chosen locus, subpopulations are homogeneous except for short periods of time right after a migration event or a mutation event. 
This suggests that if we accelerate time properly, we can neglect intra-subpopulation diversity and approximate our model by a population based model.

Inspired by these heuristics, we are going to take a low mutation - low migration regime, by making the mutation and migration rates depend on the scaling parameter $\gamma$ in the following way:
\begin{eqnarray*}
 {m_{ij}^{\gamma}}  &=&\gamma M_{ij} \ \textrm{ where $M_{ij} \ge 0$ is a constant }  \\
{b^{\gamma, \epsilon}}  &=&  \gamma  \eps \  b_{\infty}  \ \  \textrm{ where $b_{\infty}>0$ is a constant.}
\end{eqnarray*}

Recall that we take a slow mutation-migration regime but the recombination rate is constant. This is consistent with the fact that in most species, mutation rates are very low (they vary from $10^{-6}$ to $10^{-8}$ per base per generation) compared to the recombination rates (for example, for the human genome there are approximatively 66 crossovers per generation). 

\bigskip 

In a second step, we will make an additional approximation: we will consider a large population - dense site limit. In fact,  our second scaling parameter $\epsilon$, corresponds to the inverse of a typical subpopulation size. The parameters of the model depend on $\eps$ in the following way (corresponding to the second inequality in (\ref{order-parameters}): 
\begin{eqnarray*}
&& n_{i}^{\epsilon} = [N_{i}/\epsilon]   \ \textrm{ \ where $N_i>0$ remains constant as $\eps \to 0$}\\
&& l^{\epsilon} \to \infty \  \textrm{ as $\eps \to 0$}  
\end{eqnarray*}

In this article, we are going to take the limits successively: first $\gamma \to 0$ and then $\eps \to 0$, in order to be consistent  with the informal inequality \eqref{order-parameters}. We are now ready to state the main result of this paper.
\begin{theorem}
For each pair of subpopulations $i, j \in E$, let $S^i$ and $S^j$ be two independent random walks on $E$ starting respectively from $i$ and $j$ and whose transition rate from $k$ to $p$ is equal to $\tilde M_{kp}:= M_{pk}/N_k$. 
Finally, define $D_t(i,j)$ as 
\begin{equation*}
\forall t\ge 0, \ \ D_t(i,j) \ = \ 1 \ - \  \int_0^t e^{-2b_\infty s} \Pp(\tau_{ij} \in ds) \ - \ e^{-2b_\infty t} \Pp\left(\tau_{ij} > t \right),  \
\end{equation*}
where 
$\tau_{ij}=\inf\{t \geq0 \ : \ S^i(t) = S^j(t)\}.$
 
If at time $0$ the metapopulation is homogeneous (i.e. all the individuals in all subpopulations share the same genotype)
then
\begin{equation*}
\lim_{\eps \to 0} \ \lim_{\gamma \to 0} (d^{\gamma, \eps}_{t/(\gamma \eps)}(i,j), \ t \ge 0) \ = \ (D_t(i,j), \ t\ge 0)\  \textrm{ in the sense of finite dimensional distributions (f.d.d.).}
\end{equation*}
In particular, 
\begin{equation*}
\lim_{t \to \infty} D_t(i,j) \ = \ 1 - \E(e^{-2 b_\infty \tau_{ij}}). 
\end{equation*}
\label{thm-intro}
\end{theorem}

This result can be seen as a law of large numbers over the chromosome. Although the loci are linked  and they do not fix independently (and the recombination rate is constant), when considering a large number of them, they become decorrelated, regardless of the value of $\lambda$. (Note that the limiting process does not depend on $\lambda$.) The model behaves as if infinitely many loci evolved independently according to a Moran model with inhomogeneous reproduction rates (see Remark \ref{rem:33}).
The expression of the genetic distances has then a natural genealogical interpretation. $S^i$ and $S^j$ can be interpreted as the ancestral lineages starting from $i$ and $j$, and our genetic distance is related to the probability that those lines meet before experiencing a mutation (or in other words, that $i$ and $j$ are Identical By Descent (IBD)). 

\begin{rem}
In Theorem \ref{thm-intro}, we considered a rather restrictive initial condition.
In Section \ref{sect:33}, we give a stronger version of this theorem, which works for a larger range of initial conditions, but that requires to introduce several cumbersome notations. 
\end{rem}

\subsection{Consequences of our result} \label{consequence}

One interesting consequence of our result is that the genetic distance does not coincide with the classical graph distance, but instead it depends on all possible paths between $i$ and $j$ in the graph, and all the migration rates (and not only the shortest path and the direct migration rates $M_{ij}$ and $M_{ji}$), i.e., it does not only depend on the direct gene flow between $i$ and $j$ but on the whole metapopulation structure.  In particular, this suggests that  adding new subpopulations to the graph (which would correspond to colonisation of new demes), removing any edge (which could correspond to the emergence of a geographical or reproductive barrier between two subpopulations), or changing any migration rate (which could correspond to  modifying the habitat structure, for example) can potentially modify the whole genetic structure of  the population.

One striking illustration of the previous discussion is presented in Section \ref{sec:bottleneck}, where we consider an example where a geographic bottleneck is dramatically amplified in our new metric. See Figure \ref{bottleneck} and Theorem \ref{thm:ex} for a more precise statement. If we consider, as in \cite{Yama}, that two populations are different species if their genetic distance reaches a certain threshold, that will mean that this metapopulation structure promotes the emergence of two different species, each one corresponding to the population in one subgraph. Very often, parapatric speciation is believed to occur only in the presence of reduced gene flow.
Our example shows that in the presence of a geographic bottleneck, genetic differentiation is manly driven by the geographical structure of the population, i.e., even if the gene flow between two neighbouring subpopulations is approximatively identical in the graph, the genetic distance is dramatically amplified at the bottleneck (see Figure \ref{bottleneck}).

\begin{figure}
\begin{center}
\subfigure[Geographic distances]{ \includegraphics[height=2.8cm]{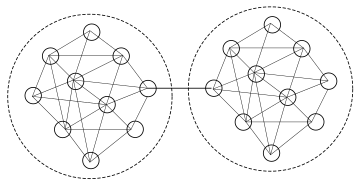}} \\
  \subfigure[Genetic distances]{\includegraphics[height=3cm]{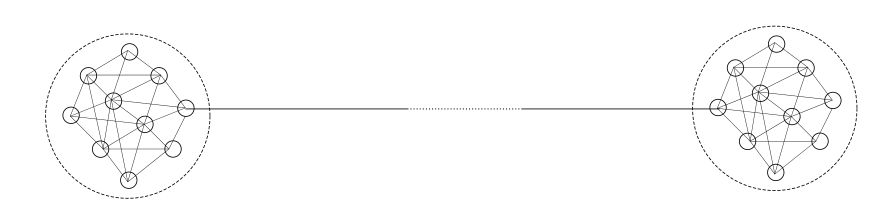}}
\caption{Amplification of a geographic bottleneck in the genetic distance metrics (small value of $c$ in Theorem \ref{thm:ex}). In this example, the metapopulation is formed of two complete graphs (not all edges are represented), connected by a single edge (a).  If $i$ and $j$ are connected, $M_{ij} = 1/d$. In (a) all the edges are the same length. In (b), the genetic distances between pairs of vertices belonging to the same subgraph are smaller than the genetic distances between pairs of vertices belonging to different subgraphs.}
\label{bottleneck}       
\end{center}
\end{figure}

\bigskip

We note that using the hitting time of random walks as a metric on graphs is not new, and has been a popular tool in graph analysis (see  \cite{doyle} and \cite{klein}). For example, the commute distance, which is the time it takes a random walk to travel from vertex $i$ to $j$ and back, is commonly used in many fields such as machine learning \citep{von2014hitting}, clustering \citep{yen}, social network analysis \citep{liben}, image processing \citep{qiu} or drug design \citep{ivanciuc, roy}.
In our case the genetic distance is given by the Laplace transform of the hitting time between two random walks, which was already suggested as a metric on graphs by \cite{Hashimoto}. In that paper the authors claimed that this metrics preserves the cluster structure of the graph. In the example alluded to above (Section \ref{sec:bottleneck}), we found that our metric reinforces the cluster structure of the metapopulation graph. In other words, a clustered geographic structure tends to increase genetic differentiation.

\subsection{Discussion and open problems}
As already mentioned above, the main result is obtained by: (i) proving that, in a low mutation - low migration regime (i.e., when $\gamma \to 0$), subpopulations are monomorphic most of the time and our individual based model converges to a population based model, (ii) showing that, under a large population - dense site limit (i.e. taking $\eps \to 0$), the genetic distances between subpopulations (for the population based model) converge to a deterministic process (defined in Theorem \ref{thm-intro}). 
Taking these two limits successively gives no clue on how the parameters should be compared to ensure the approximation to be correct. It would be interesting to take the limits simultaneously but it is technically challenging (for example we would need to characterise the time to fixation for $l$ loci that do not fix independently, which is not easy).

As discussed in the previous paragraph, we can only show our results under some rather drastic constraints: subpopulations are asymptotically monomorphic.
More generally, we believe that Theorem \ref{thm-intro} should hold under relaxed assumptions, namely when the intra-subpopulation genetic diversity is low compared to the inter-subpopulation diversity (see Figure \ref{simu} for an example, where $\gamma = 2e^{-6}$ and $\eps = 5e^{-3}$ ). Technically, this would correspond to the condition that at a typical locus (i.e, a locus chosen uniformly at random) each subpopulations is monomorphic at that site with high probability (which is in essence  \eqref{order-parameters}). Of course, proving such a result would be much more challenging, but would presumably correspond to a more realistic situation.

\begin{figure}
\begin{center}
\includegraphics[width=10cm]{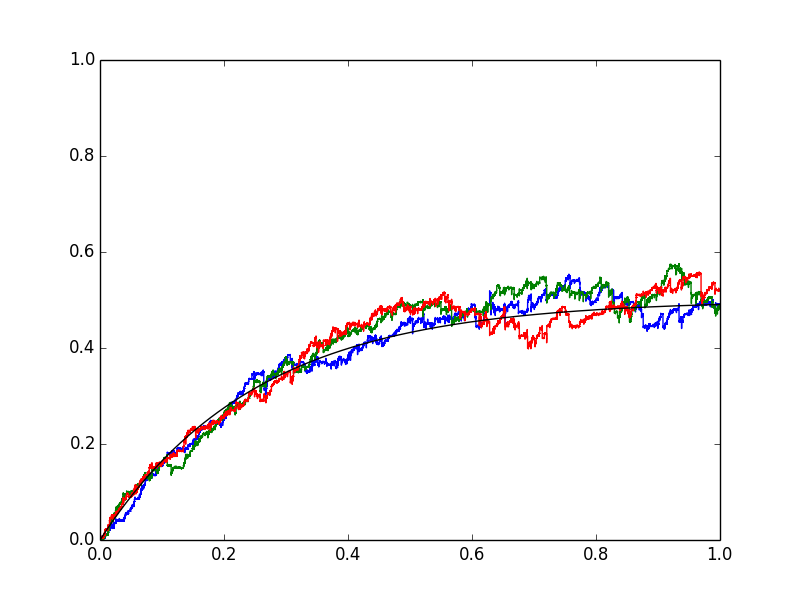}
\caption{Simulation of the individual based model, for $d = 3$, $N_1 = N_2 = N_3 = 1$, $\eps$ = 0.005, $\gamma = 2e^{-6}$, $l^\eps = 100$, $\lambda = 10$.
 The black curve corresponds to $D_t(i,j)$ (see Theorem \ref{thm-intro}). The blue, green and red curves correspond to the three genetic distances $d_t^{\eps, \gamma}(1,2), d_t^{\eps, \gamma}(2,3), d_t^{\eps, \gamma}(1,3)$.}
\label{simu}
\end{center}
\end{figure}

\subsection{Outline}  In Section \ref{pbm}, we show that in the rare mutation-rare migration regime (i.e. when $\gamma \to 0$ whereas $\eps$ remains constant), the individual based model (IBM) described above converges to a  population based model (PBM) (see Theorem \ref{IBMtoPBM}). 
This PBM is a  generalization of the model proposed by \citet{Yama, Yama2} in three ways. First, it is an extension of their model from two to an arbitrary number of subpopulations, which is not trivial from a mathematical point of view. 
Second, in \cite{Yama, Yama2}, the authors only assumed that the migrant alleles are fixed independently at every locus. To make the model more realistic, we took into account genetic linkage, which introduces a non-trivial spatial correlation between loci (along the chromosome). Finally, we suppose that the loci are distributed randomly along the chromosome (and not in a regular fashion). Section \ref{pbm} is interesting on its own since it provides a theoretical justification of the model proposed by \citet{Yama, Yama2}.

In Section \ref{mainresult} and \ref{sect:proof}, we study the PBM in the large population - dense site limit (i.e. when $\eps \to 0$). We properly introduce the main tool used to study the population based model -- the genetic partition probability measure -- and show an ergodic theorem related to this process (see Theorem \ref{thm31}).

Finally, in Section \ref{sect:33} we prove our main result (Theorem \ref{thm-main2} which is an extension of Theorem \ref{thm-intro})  by combining the results of the previous sections.

Section \ref{sec:bottleneck} proves the result  related to the geographical bottleneck alluded to in Section \ref{consequence} (see Proposition \ref{sec:bottleneck}).

\section{Approximation by a population based model}
\label{pbm}

\label{popbased}
We now describe  a population based model (PBM) that can be seen as the limit of the IBM presented above, when $\gamma$ goes to $0$ (whereas $\eps$ remains fixed) and time is rescaled by $1/(\gamma\eps)$. 
 Consider a metapopulation where the individuals are characterised by a finite set of loci, whose positions are distributed as $l^\eps$ uniform random variables on $[0,1]$, and let ${\cal L}^\eps$ be the vector of the positions of the loci  (as described in Section \ref{sub-IBD} for the IBM).
We  now describe the dynamics of the model, conditional on ${\cal L}^\eps = L^\eps$, with $L^\eps \in [0,1]^{l^\eps}$.

Before going into the description of our model, we start with a definition.
It is well known that the Moran model  reaches fixation in finite time, i.e., after a (random) finite time, every individual in the population carries the same genetic material, and from that time on, the system remains trapped in this configuration (see \cite{Kimura1}, \cite{Kimura2}). 
 
 \begin{definition}\label{def-F}
Consider a single population of size $n_j^\eps$ formed by a mutant individual (the migrant) and $n^{\eps}_j-1$ residents, that evolves according to a Moran model with recombination at rate $\lambda$ (as described in Section \ref{sub-IBD}). We define ${\cal F}^{L^\eps, \lambda}_{j}$ as the (random) set of loci carrying the mutant type when the population becomes homogeneous. (Note that ${\cal F}^{L^\eps, \lambda}_{j}$ is potentially empty.)
\end{definition}

We are now ready to describe our PBM. We represent each subpopulation as a single chromosome, which is itself represented by the set of loci $L^\eps$.  The dynamics of the population can then be described as follows.
\begin{itemize}
\item For every $i\in E$: fix a new mutation in population $i$ at rate $b_\infty l^\eps$, the locus being chosen uniformly at random along the chromosome.
\item For every $i,j\in E$ and every $S\subseteq\{1,\dots,l^\eps\}$: at every locus in $S$, fix simultaneously  the alleles from population $i$ in population $j$ at rate $\frac1{\eps} M_{ij} {\mathbb P}\left({\cal F}^{L^\eps,  \lambda}_{j}= S\right)$.
\end{itemize}

In the PBM (parametrised by $\eps$), we define the genetic distance between subpopulations $i$ and $j$ at time $t$ as follows:
$$d^{\epsilon}_t({i,j}) \ = \ \frac{1}{l^\eps}\#\{ \ k\in\{1,\dots,l^\eps\} \ : \ \mbox{subpopulations $i$ and $j$ differ at locus $k$ at time $t$}  \ \}$$
as opposed to $d^{\gamma,\eps}$  which will refer to the genetic distances in the IBM as described in Section \ref{sub-IBD} (parametrised by $\gamma$ and $\eps$).
We note that the definition of the genetic distance in the PBM is consistent with the one in the IBM (see (\ref{def:genetic-distance})) in the sense that if the subpopulations are homogeneous in the IBM, (\ref{def:genetic-distance}) is equal to the RHS of the previous equation.
We are now ready to state the main result of this section.

\begin{theorem}
Assume that, at time $0$, the subpopulations in the IBM are homogeneous and that $\forall i,j \in E$, $d^{\gamma, \eps}_0(i,j) = d^{ \eps}_0  (i,j)$. 
Then, for every $k \in \N$, $\forall \  0 \leq t_1<\dots<t_k$, 
\begin{equation}
\lim \limits_{\gamma \to 0}  \ (d^{\gamma, \eps}_{t_1/(\gamma\eps)}, \dots, d^{\gamma, \eps}_{t_k/({\gamma \eps})})  \ = \ (d^{ \eps}_{t_1}, \dots, d^{ \eps}_{t_k}) \ \textrm{ in distribution.}
\label{couplage}
\end{equation}
\label{IBMtoPBM}
\end{theorem}

\begin{proof}
Recall that the loci are distributed randomly along the chromosome. In the proof, we assume that the vector of the positions of the loci ${\mathcal L}^\eps$ is fixed and equals to $L^\eps \in [0,1]^{l^\eps}$ (and is the same in the IBM and in the PBM). We also consider that IBM and the PBM start from the same deterministic initial condition.
The unconditional extension of the proof can be easily deduced from there.

We define a coupling between the IBM and a new PBM that is close (in distribution) to the PBM defined at the beginning of this section.  The idea behind the coupling is that, when time is accelerated by $1/(\gamma \eps)$, and $\gamma$ is small,  in the IBM, the time to fixation after a mutation or migration event is short enough so that the population has become homogeneous before the next  mutation or migration event takes place. Then, we can decompose the trajectories of the IBM into periods where the population is homogeneous (and waits for the next mutation or migration event to take place) and homogenization phases (where the dynamics of the population is described by a Moran model). See Figure \ref{coup} for an illustration of this concept.

\begin{figure}
\begin{center}
\subfigure[]{\includegraphics[width=8cm]{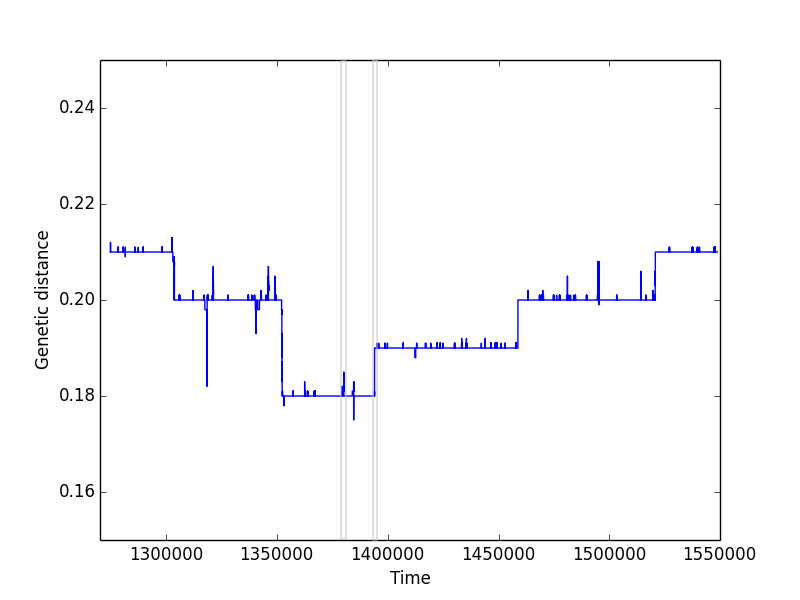}}
\subfigure[]{\includegraphics[width=8cm]{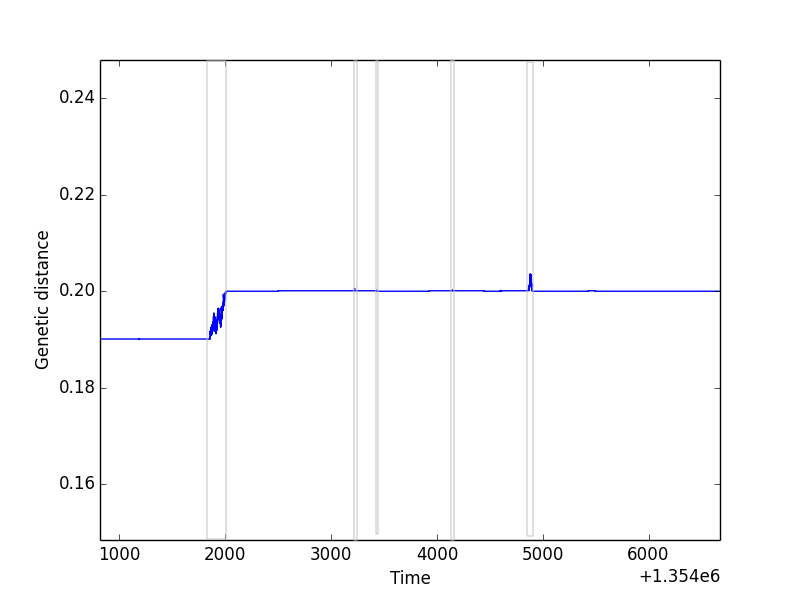}}
\caption{(a) Simulation of the individual based model, for $d = 3$, $N_1 = N_2 = N_3 = 1$, $\eps$ = 0.005, $\gamma = 2e^{-6}$, $l^\eps = 100$, $\lambda = 10$. It is the same simulation as in Figure \ref{simu} (without rescaling time).
Two examples of an homogenization phases are contained in the grey boxes.  (b)  Partial zoom on the curve.}
\label{coup}
\end{center}
\end{figure}

More formally, let us consider $(Y_t^{\gamma,\eps}; t\geq0)$ the process recording the genetic composition in the IBM (i.e. a matrix containing the sequences of the chromosomes of all the individuals in the metapopulation)
{\it after rescaling time by $\gamma \eps$} so that
\begin{enumerate}
\item  For $i \in E$, mutation events on the subpopulation $i$ occurs according to a Poisson Point Process (PPP) with intensity measure $b_\infty l^{\eps} n_i^{\eps} dt$.
\item For $i,j \in E$, migration events from $i$ to $j$ can be described in terms of a PPP with intensity measure $M_{ij}/\eps dt$.
\end{enumerate} 
Define ${\mathcal E}^{\gamma,\eps}$ the event that every time  a subpopulation is affected by a mutation or a migration event on the interval $[0,T]$, the subpopulation is genetically homogeneous when the event occurs (as in Figure \ref{coup}). In other words, there is no overlap between mutation and migration homogenization periods. The time to fixation in our (multi-locus) Moran model only depends on the number of individuals and the number of loci, so in our model it only depends on $\eps$ (but not on $\gamma$). As a consequence, ${\mathbb P}\left({\mathcal E}^{\gamma,\eps}\right)\to 1$ as $\gamma \to 0$.

Next, let us consider $T_t^{\gamma,\eps}$ as the Lebesgue measure of 
\[\{0\leq s\leq t\ : \ \forall i\in \ E \mbox{ pop. $i$ is homogeneous at time $s$}   \}.\]
In words, $( T_t^{\gamma,\eps}; t\geq 0)$ is the random clock which is obtained by skipping the homogenization period after a migration or mutation event (i.e. by skipping the red intervals in Figure \ref{coup}).
By arguing as in the previous paragraph, as $\gamma \to 0$, it is not hard to see that $(T^{\gamma,\eps}_t; t\geq0)$ converges to the identity in the Skorokhod topology on $[0,T]$
for every $T\geq0$.

Let us now consider 
\[Z^{\gamma,\eps}_t \ =\  Y^{\gamma,\eps}_{(T^{\gamma,\eps})^{-1}_t}, \ \mbox{where } \ (T^{\gamma,\eps})^{-1}_t \ = \ \inf\{s\geq0\ : \  T^{\gamma,\eps}_s \geq t  \}. \]
By construction, this process defines a PBM in the sense that at every time $t$, any subpopulation is composed by genetically homogeneous individuals.
Further, since $(T^{\gamma,\eps}_t; t\geq0)$ converges to the identity and mutation and migration events occur at Poisson times, the finite  dimensional distributions of $Z^{\gamma,\eps}$ are  good approximations of the ones for the IBM.

Let us now show that $Z^{\gamma,\eps}$ (constructed from the IBM) is close in distribution to the PBM defined at the begining of  this section. Conditioned on the event ${\cal E}^{\gamma,\eps}$ (whose probability goes to $1$) and on the PPP's described in 1 and 2 above, the PBM $Z^{\gamma,\eps}$ can be described as follows. Define $p_{\Delta t,i}$ to be the probability for a mutant allele (at a given locus of a given individual) to fix in a population of size $n_i^\eps$, {\it  conditioned on the homogenization time to be smaller than $\Delta t$}. Then the distribution of the conditioned PBM $Z^{\gamma,\eps}_t$  can be generated as follows.
\begin{enumerate}
\item[(a)] At every mutation time $t$ in subpopulation $i\in E$, choose a locus $k$ uniformly at random and fix the mutation instantaneously with probability $p_{\frac{\Delta t}{\gamma \eps},i}$, where $\Delta t$ is the time between $t$ and the next mutation or migration event (in our new time scale). We note that if the mutation does not fix, then $Z^{\gamma,\eps}$ is not affected by the mutation event, and as a consequence ``effective mutation'' events in $Z^{\gamma,\eps}$ are obtained from the mutation events in the IBM after thinning each time with their respective probability $p_{\frac{\Delta t}{\gamma \eps},i}.$
\item[(b)] At every migration event $t$ on subpopulation $j$, fix a random set $S$ where $S$ is chosen according to ${\mathcal F}_j^{L^\eps,\lambda, \Delta t/(\gamma\eps)}$,
where $\Delta t$ is defined as in the previous point, and where    ${\mathcal F}_j^{L^\eps, \lambda,s}$ is the random variable ${\mathcal F}_j^{L^\eps, \lambda}$ conditioned on the homogenization to occur in a time smaller smaller than $s$.
\end{enumerate}
Since fixation (of one of the alleles) occurs in finite time almost surely, and  the distribution of the homogenization time only depends on $\eps$, we have
\[ {\mathcal F}_j^{L^\eps, \lambda,\Delta t/(\gamma\eps)} \underset{\gamma\to 0}{\Longrightarrow}  {\mathcal F}_j^{L^\eps, \lambda}, \ \mbox{and} \  \  \lim_{\gamma\to 0} p_{\Delta t/(\gamma\eps),i} \ = \ \frac{1}{n_j^\eps} \]
where the RHS of the second limit is the probability of fixation of a mutant allele in the absence of conditioning.

\medskip

Putting all the previous observations together, one can easily show that the genetic distance in $Z^{\gamma,\eps}$ converges (in the finite dimensional distributions sense) to the ones of the PBM. In particular, we recover the mutation rate on subpopulation $i$ in the PBM
\[ \underbrace{b_\infty l^{\eps} n_i^{\eps} }_{\mbox{rate of mutation in the IBM}}\times \underbrace{\frac{1}{n_i^\eps}}_{\mbox{proba of fixation}} =  b_\infty l^\eps,\] 
which corresponds to the limiting ``effective mutation rate'' in the PBM $Z^{\gamma,\eps}$ (see (a) above) as $\gamma\to 0$.
This completes the proof of Theorem \ref{IBMtoPBM}.
\end{proof}

We note that Theorem \ref{IBMtoPBM} could be extended to the case where the subpopulations are not homogeneous in the IBM at time $t=0$.
 Indeed, arguing as in the proof of Proposition \ref{IBMtoPBM}, if we start  with some non-homogeneous initial condition, then each island becomes homogeneous before experiencing any mutation or mutation event with very high probability. In order to get an efficient coupling between the IBM and the PBM, we simply choose the initial condition of the PBM as the (random) state of the IBM after this initial homogenization period.

\begin{rem}\label{rem:single-locus}
Choose a locus $k\in\{1,\dots,l^{\eps}\}$. We let the reader convince herself that in the PBM the genetic composition at a given locus $k$ follows the following Moran-type dynamics:
\begin{enumerate}
\item[(mutation)] ``Individual'' $i$ takes on a new type (or allele) at rate $b_{\infty}$. 
\item[(reproduction)] ``Individual'' $j$ inherits its type from ``individual'' $i$ at rate 
$
(1/{\eps})M_{ij}\Pp\left(k\in{\cal F}^{L^\eps,  \lambda}_{j}\right).
$
Further, in a neutral one-locus Moran model, the probability of fixation of a single allele in a resident population of size $n_j^\eps$ is equal to its initial frequency, which in our case is $1/n_j^\eps$. Thus 
$$
\frac1{\eps} M_{ij}\Pp\left(k\in{\cal F}^{L^\eps, \lambda}_{j}\right) \ =  \ \frac1{\eps} \ M_{ij}/n_{j}^\eps.
$$
This dynamics is not dependent on the position of the locus under consideration.
\end{enumerate}
\end{rem}

\begin{rem}\label{rem:33}
Our model can be seen as a multi-locus Moran model with inhomogeneous reproduction rates. The main difficulty in analysing this model stems from the fact that there exists a  non trivial correlation between loci. This correlation is induced by the fact that fixation of migrant alleles can occur simultaneously at several loci
during a given migration event.  In turn, the set of fixed alleles during a given migration event is determined by the local  Moran dynamics described in the Introduction. 
\end{rem}

\section{Large population - dense site limit}
\label{mainresult}
In this Section, we study the PBM described in the Section \ref{popbased}, in the large population - dene site limit. In particular, we study the dynamics of the genetic distances and we state  Theorem \ref{thm31}, that together with Theorem \ref{IBMtoPBM}  implies the main result of this article, namely Theorem \ref{thm-main2}, that is a stronger version of Theorem \ref{thm-intro}  (see Section \ref{sect:33}).

\subsection{The genetic partition measure}
\label{sect:idea-proof}
The main difficulty in dealing with the genetic distance is that it lacks the Markov property, and as a consequence, it is not directly amenable to analysis. In fact, when $d >2$, a migration event from $i$ to $j$ can potentially have an effect on the genetic distance between $j$ and another subpopulation $k$ (see Figure \ref{fig3} for an example).

\begin{figure}[h]
\begin{center}
\includegraphics[width=5cm]{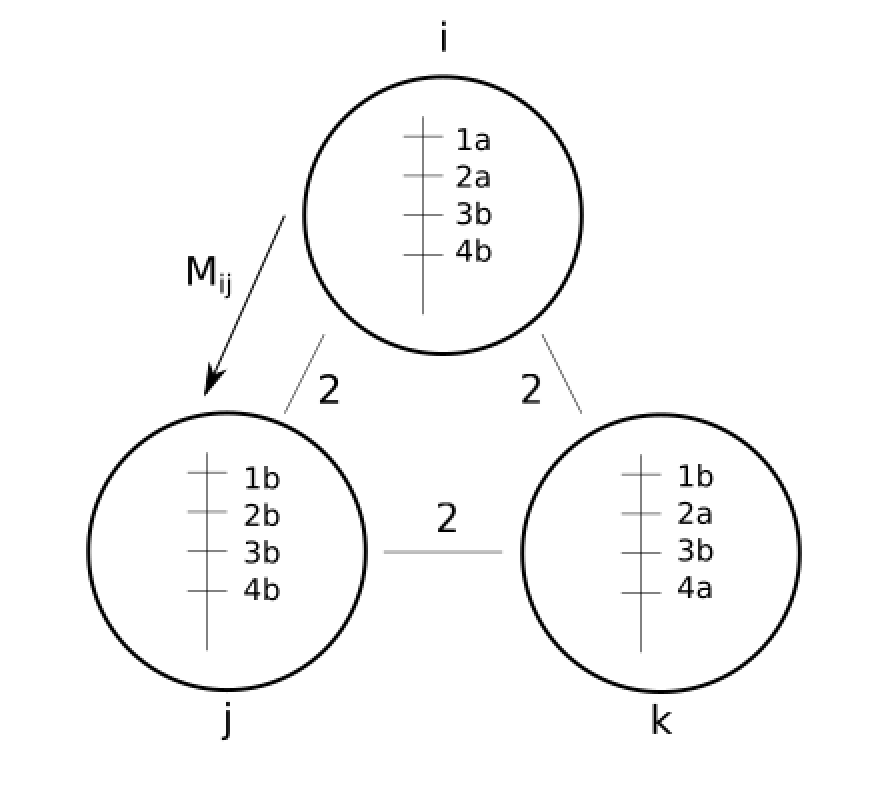} \ \ \ \ 
\includegraphics[width=5cm]{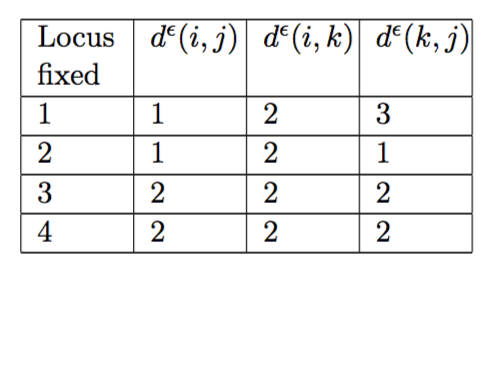} 
\caption{The three subpopulations (i,j,k) are characterised by a chromosome with four loci (1,2,3,4), with different alleles (1a, 1b, ...). The three genetic distances ($d^\eps(i,j), d^\eps(i,k), d^\eps(k,j)$) are equal to two (before migration). The table shows the new genetic distances  after a migration event from $i$ to $j$ where one locus from $i$ is fixed in population $j$.  At locus $1$, the allelic partition $\Pi^\eps_1(t)$ is equal to $\{i\}\{j,k\}$, whereas at locus $4$, $\Pi^\eps_4(t)=\{i,j\} \{k\}$.}
\label{fig3}
\end{center}
\end{figure}

To circumvent this difficulty, we now introduce an auxiliary process -- the genetic partition probability measure  -- from which one can easily recover the genetic distances (see \eqref{eq:pmgp-to-distance} below), and whose asymptotical dynamics is explicitly characterised in Theorem \ref{thm31} below. 

Let ${\cal P}_d$ be the set of partitions of $\{1, \dots, d\}$. Fix $\pi\in {\cal P}_d$ and $i,j\in E$. Define $\mathcal{S}_{i}(\pi)$ as the element of ${\cal P}_d$ obtained from $\pi$ by making $i$ a singleton (e.g., ${\cal S}_2(\{1,2,3\})=\{1,3\}\{2\}$). Define ${\cal I}_{ij}(\pi)$ as the element of ${\cal P}_d$ obtained from $\pi$ by displacing $j$ into the block containing $i$ (e.g., ${\cal I}_{2,3}(\{1,3\}\{2\})=\{1\}\{2,3\}$).

At every locus $k\in\{1,\dots,l^{\eps}\}$ (ordered in increasing order along the chromosome) and every time $t$, the allele composition of the metapopulation induces a partition on $E$.
More precisely, at locus $k$, two subpopulations are in the same block  of the partition at time $t$ iff they share the same allele  at locus $k$.

In the following, fix $L^\eps \in [0,1]^{l^\eps}$, the vector containing the positions of the loci.
In the PBM parametrised by $\eps$, we condition on the loci being located at $L^\eps$ and 
for every $k\in\{1,\dots,l^{\eps}\}$ we let $\Pi^{\eps, L^\eps}_k(t)$ be the partition induced at locus $k$ (see Figure \ref{fig3}). 
The vector $\Pi^{\eps, L^\eps}(t)=\left(\Pi^{\eps, L^\eps}_k(t);  \ k\in\{1,\dots,l^{\eps}\}\right)$ describes the genetic composition of the population at time $t$.  According to the description of our dynamics, $\Pi^{\eps, L^\eps}(t)$  is a Markov chain with the following transition rates:
 
\begin{itemize}
\item(mutation) For every $i\in E$,  $k\in\{1,\dots,l^{\eps}\}$, define ${\cal S}_i^k$ to be the operator on $(\cP_d)^{l^{\eps}}$ such that $\forall \ \Pi \ \in (\cP_d)^{l^{\eps}}$
\begin{equation*}
{\cal S}_i^k(\Pi) = (\tilde \Pi_1,\ldots,\tilde \Pi_{l^\epsilon}) \ \textrm{ with }  \left\{
      \begin{aligned}
        & \tilde \Pi_j = \Pi_j\  \ \ &\forall j \ne k \\
        & \tilde \Pi_k = {\cal S}_i(\Pi_j)\  \ \ & j = k  \\
      \end{aligned}
    \right. \  .
\end{equation*}

The transition rate of the process from state $\Pi$ to ${\cal S}_i^k(\Pi)$ is given by  $b_{\infty}$. 

\item(migration from $i$ to $j$) For every $i,j\in E$,  $S\subset\{1,\dots,l^{\eps}\}$, define ${\cal I}_{ij}^S$ the operator on $(\cP_d)^{l^\eps}$ such that $\forall \ \Pi \ \in (\cP_d)^{l^{\eps}}$
\begin{equation*}
{\cal I}_{ij}^S(\Pi) = (\tilde \Pi_1,\ldots,\tilde \Pi_{l^\epsilon}) \ \textrm{ with } \left\{
      \begin{aligned}
        & \tilde \Pi_j = \Pi_j\   \ \ &\forall k \notin  S  \\
        &  \tilde \Pi_k =  {\cal I}_{ij}(\Pi_k) \  \ \ &\forall k \in S  \\
      \end{aligned}
    \right. .
\end{equation*}

The transition rate of the process from $\Pi$ to ${\cal I}_{ij}^S(\Pi)$ is given by $\frac{M_{ij}}{\eps} {\mathbb P}\left({\cal F}^{L^\eps, \lambda}_{j}=S\right)$.
\end{itemize}
To summarise, for any function $h: {\cal P}_d^{l^\epsilon} \to \R$, the generator of $\Pi^{\eps,L^\eps}$ can be written as 
\beqnn
{\mathbb G}^{\eps, L^\eps} h(\Pi) & = &  \frac{1}{\eps}\sum_{i,j=1}^d  M_{ij} \sum_{S\subset \{1,\dots,l^\eps\}} \Pp({\cal F}^{L^\eps, \lambda}_{j}= S) \ [h({\cal I}_{ij}^S(\Pi)) - h(\Pi)  ]  \ + \nonumber \\
&&  b_{\infty} \sum_{i=1}^d \sum_{k=1}^{l^\eps}  \left(   h({\cal S}_i^k(\Pi))-h(\Pi)\right).  \label{eq:generator}
\eeqnn

\subsection{Some notation}

Let ${\cal M}_{d}$ denote the space of signed finite measures on ${\cal P}_d$. Since ${\cal P}_d$ is finite, we can identify elements of ${\cal M}_{d}$ as vectors of $\R^{Bell_d}$, where $Bell_d$ is the Bell number, which counts the number of elements in ${\cal P}_d$ (the number of partitions of  $d$ elements). In particular, if $\pi$ is a partition of $E$, and $\mu\in{\cal M}_d$, then $\mu(\pi)$ will correspond to the measure of the singleton $\{\pi\}$, or equivalently, to the  ``$\pi^{th}$ coordinate'' of the vector $\mu$. 
We define the inner product $<\cdot, \cdot>$ as
 \begin{eqnarray*} 
 <\cdot, \cdot> :  \  {\cal M}_d \times {\cal M}_d &\rightarrow& \R \nonumber \\
 m,v &\rightarrow&  \sum_{\pi \in {\cal P}_d} m(\pi)v(\pi). \nonumber 
 \end{eqnarray*} 
 For every function $f: {\cal P}_d \rightarrow {\cal P}_d$, we define the operator $*$ s.t for every $m \in {\cal M}_d$, for every $\pi \in {\cal P}_d$, $f*m(\pi) = m(f^{-1}(\pi))$.
In words, $f*m$ is the push-forward measure of $m$ by $f$. 
Further, we will also consider square matrices indexed by elements in $\cP_d$. For such a matrix $K$ and an element $m\in{\cal M}_d$, we define $Km(\pi) := \sum_{\pi'\in {\cal P}_d} K(\pi,\pi')m(\pi')$.

Define 
\begin{eqnarray*}
X & : &  (\cP_d)^{l^\eps}  \ \rightarrow  \  \cM(\cP_d) \nonumber \\
   &  & \Pi   \ \rightarrow  \ \frac1{l^\eps} \ \sum_{k\leq l^{\eps}} \delta_{\Pi_k},
\end{eqnarray*}
i.e., $X(\Pi)$ is the empirical measure associated to the ``sample'' $\Pi_1,\dots,\Pi_{l^\eps}$.
In the following, we define
\begin{equation*}
\xi^{\eps,L^\eps}_t \ := \ X(\Pi^{\eps, L^\eps}(t)) 
\end{equation*}
will be referred to as the (empirical) genetic partition probability measure of the population, conditional on the $l^\eps$ loci to be located at $L^\eps$.
We also define 
\begin{equation*}
\xi^{\eps}_t \ \equiv \ \xi^{\eps, {\cal L}^\eps}_t  \ =  \ X(\Pi^{\eps, {\cal L}^\eps}(t)) \ \textrm{where $ {\cal L}^\eps \sim {\cal U}([0,1]^{l^{\eps}})$}
\end{equation*}
will be referred to as the (empirical) genetic partition probability measure of the population. 

\bigskip

The genetic distance between $i$ and $j$ at time $t$ can then be expressed in terms of  $\xi^\eps_t$ as follows:
\begin{equation} \label{eq:pmgp-to-distance}
d_t^\eps(i,j) = 1-\xi^{\eps}_t( \{\pi \in {\cal P}_d : \ i\sim_\pi j)  \} ).
\end{equation}

In the following, we identify the process $(\xi^{\eps}_t, t\ge0)$ to a process in the set of the c\`adl\`ag functions from $\R^+$ to $\R^{Bell_d}$, equipped with the standard Skorokhod topology.

\subsection{ Convergence of the genetic partition probability measure}
\label{sect:32}
Following Remark \ref{rem:single-locus}, for every $k\in\{1,\dots,l^{\eps}\}$,
the process
$(\Pi_k^{\eps,{L}^\eps}(t); t\geq0)$ -- the partition at locus $k$ -- obeys the following dynamics:
\begin{enumerate}
\item(reproduction event)  $j$ is merged in the block containing  $i$ at rate $M_{ij} \frac1{ \eps n^{\eps}_j}$.
\item(mutation) Individual $i$ takes on a new type at rate $b_{\infty}$.
\end{enumerate}
The generator associated to the allelic partition at locus $k$ is then given by
\be\label{def:G-eps}
G^\eps g(\pi) = \sum_{i,j=1}^d M_{ij}\frac{1}{\eps n^{\eps}_j}\left( \ g({\cal I}_{ij}(\pi)) - g(\pi)  \  \right) \ + \ b_{\infty} \sum_{i=1}^d \left( \ g({\cal S}_{i}(\pi)) - g(\pi)  \  \right).
\ee
Recall that the expression of the generator associated to the allelic partition at a given locus $k$ is independent on the position of locus $k$ and on $\lambda$.
Also recall that $\eps n_i^\eps \to N_i$ and that, by definition, $\forall k,p \in E, \ \tilde M_{kp} = M_{pk}/N_k$. Thus
\be\label{def:G-inf}
G^\eps g(\pi) \to G g(\pi) := \sum_{i,j=1}^d \tilde M_{ji} \left( \ g({\cal I}_{ij}(\pi)) - g(\pi)  \  \right) \ 
+ \ b_\infty \sum_{i=1}^d \left( \ g({\cal S}_{i}(\pi)) - g(\pi)  \  \right) \  \ \mbox{as $\eps\to 0$.}
\ee
Direct computations yield that $^t G$, the transpose of the matrix $G$ satisfies
\begin{equation}
\forall m\in{\cal M}_d, \ \ ^tGm  \ =  \  \sum_{i,j=1}^d \tilde M_{ji}  ({\cal I}_{ij}*m - m) + {b_\infty} \sum_{i=1}^d \left({\cal S}_i*m - m\right). 
\label{generatorG}
\end{equation}

In the light of (\ref{def:G-inf}), the following theorem can be interpreted as an ergodic theorem. We show that the (dynamical) empirical measure constructed from the allelic partitions along the chromosome converges to the probability measure of a single locus. Although in the IBM the different loci are linked and do not fix independently (as already mentioned in Remark \ref{rem:33}), as the number of loci tends to infinity, they become decorrelated. In the large population - dense site limit, the following result indicates that the model behaves as if infinitely many loci evolved independently according to the (one-locus) Moran model with generator $G$ provided in (\ref{def:G-inf}).  

In the following $``\Longrightarrow''$ indicates the convergence in distribution. Also, we identify $(\xi_t^\eps; \ t\ge0)$ to a function from $\R^+$ to $\R^{Bell_d}$; and convergence \textit{in the weak topology} means that for every $T>0$, the process $(\xi_t^\eps; \ t\in[0,T])$ converges in the Skorohod topology $D([0,T], \R^{Bell_ d})$.

\begin{theorem}[Ergodic theorem along the chromosome] 
Assume that $\xi^\eps_0$ is deterministic and there exists a probability measure $P^0 \in {\cal M}_d$ such that the following convergence holds: 
\begin{equation}\label{initial}
\xi^\eps_0 \underset{\eps \to 0}{\longrightarrow} P^0.
\end{equation}
Then 
\begin{equation*}
(\xi_t^\eps; \ t\ge0) \underset{\begin{subarray}{l} \epsilon \to 0 \end{subarray} }{\Longrightarrow} (P_t; \ t\ge0) \ \mbox{in the weak topology},\end{equation*}
where $P$ solves the forward Kolmogorov equation associated to the aforementioned Moran model, i.e.,
\begin{equation*}
\frac{d}{ds}P_s \ = \  ^tG P_s 
\end{equation*}
with initial condition $P_0=P^0$ and where $^t G$ denotes the transpose of $G$ (see (\ref{generatorG})). 
\label{thm31}
\end{theorem}

\section{Proof of Theorem \ref{thm31}}
\label{sect:proof}

The idea behind the proof is to condition on ${\cal L}^\eps= L^\eps$, and then decompose the Markov process 
$(<\xi^{\eps, L_\eps}_t,v>; t\geq0)$ into a drift part and a Martingale part.
We show that the drift part converges to the solution of the Kolmogorov equation alluded to in Theorem \ref{thm31} and  
  that the Martingale part vanishes when $\eps \to 0$.
 The main steps of the computation are outlined in the next subsection. We leave technical details (tightness and second moment computations)
 until the end of the section.

\subsection{Main steps of the proof}

Fix $L^\eps \in [0,1]^{l^\eps}$. Recall the definition of ${\mathbb G}^{\eps, L^\eps}$, the generator of the process $(\Pi^{\eps,L^\eps}(t); \ t\geq0)$, given in \eqref{eq:generator}. Let $f: \R \to \R$ be a Borel bounded function and fix $v \in {\cal M}_d$. Define $h(\Pi) \ = \ f(\left<X(\Pi),v\right>)$. Then, it is straightforward to see from   \eqref{eq:generator} that
\begin{eqnarray}
{\mathbb G}^{\epsilon, L^\eps} h(\Pi) &  = & \frac{1}{\epsilon}  \sum_{i,j=1}^d  M_{ij}  \  \E_{\lambda,L^\eps,j} \left(   f(\left<X({\cal I}_{ij}^{ S}(\Pi)),v\right>)-f(\left<X(\Pi),v\right>) \right)  \nonumber \\
& + &  b_{\infty} l^{\eps}  \sum \limits_{i=1}^d  \  \E_{l^\eps}\left(   f(\left<X({\cal S}_i^K(\Pi)),v\right>)-f(\left<X(\Pi),v\right>) \right)  , 
\label{eq:gener-x}
\end{eqnarray}
where in the first line  $ \E_{\lambda,L^\eps,j}$ is the expected value  taken with respect to the random variable $S$, distributed as ${\cal F}_{j}^{L^\eps, \lambda}$
as defined in Definition \ref{def-F}.
In the second line, $ \E_{l^\eps}$ is the expected value is taken with respect to $K$, distributed as a uniform random variable on $\{1,\dots,l^{\eps}\}$.

\begin{lemma}\label{lem:gen-first-order}
Define $v \in {\cal M}_d $, $L^\eps \in [0,1]^{l^\eps}$, $ g(\Pi) := \left<X(\Pi), v\right>$. Then ${\mathbb G}^{\eps, L^\eps} g(\Pi) \ = \ \left<^t G^\eps X(\Pi), v\right>$ where $^t G^\eps$ is the transpose of $G^\eps$ -- the generator of the allelic partition at a single locus as defined in (\ref{def:G-eps}) -- i.e.,
\begin{equation}
\forall m \in{\cal M}_d, \ \  ^t G^\eps m \ = \ \sum_{i,j=1}^d M_{ij} \frac{1}{\eps n^{\eps}_i} ({\cal I}_{ij}*m - m) + b_{\infty} \sum_{i=1}^d \left({\cal S}_i*m - m\right).
\label{Gepsdef}
\end{equation}
\end{lemma}
\begin{proof}
We define the two following signed measures:
\begin{equation}\label{defdronde}
\partial_{ij}^{ S} X(\Pi) \ = \ X({\cal I}_{ij}^{S}(\Pi)) \ - \ X(\Pi), \ \  \partial_i^K X(\Pi) \ = \ X({\cal S}_i^K(\Pi)) \ - \ X(\Pi).
\end{equation}
In words, $\partial_{ij}^{ S} X(\Pi)$ is the change in the genetic partition measure $X(\Pi)$ if we merge $j$ in the block of $i$ at every locus in $S$, and $\partial_i^K X(\Pi)$  is the change in $X(\Pi)$ if we single out element $i$ at locus $K$.
Using those notations, for our particular choice of $g$, \eqref{eq:gener-x} writes
\begin{eqnarray*}
{\mathbb G}^{\epsilon, L^\eps} g(\Pi) &  = & \frac{1}{\eps}\sum_{i,j=1}^d  M_{ij} \  \E_{\lambda,L^\eps,j} \left(   \left< \partial_{ij}^{S} X(\Pi),v\right>  \right) 
\ + \  b_{\infty}l^\eps  \sum_{i=1}^d  \  \E_{l^\eps}\left( \left<  \partial_i^K X(\Pi), v \right>\right).  
\end{eqnarray*}
We now show that for every $v \in {\cal M}_d$, 
\begin{eqnarray}
  \E_{\lambda,L^\eps,j} \left(  \left<  \partial_{ij}^{S} X(\Pi), v\right>  \right) \   &=& \  \frac{1}{n^{\eps}_j} \left<{\cal I}_{ij}*X(\Pi) - X(\Pi), v\right>,  \nonumber \\
  \E_{l^\eps} \left(    \left<  \partial_i^K X(\Pi), v \right>  \right)   \ &=& \   \frac{1}{l^{\eps}}\left<{\cal S}_{i}*X(\Pi) - X(\Pi), v \right>. \label{eq:exp-partial}
\end{eqnarray}
We only prove the first identity. The second one can be shown along the same lines.
Again, we let $\Pi_k$ be the $k^{th}$ coordinate of $\Pi$.  By definition, the vector ${\cal I}_{ij}^{S}(\Pi)$
is only modified at the coordinates belonging to $S$, and thus
\begin{eqnarray}
\left<\partial_{ij}^S X(\Pi),v\right> 
& = & \sum \limits_{\pi \in \mathcal{P}_d} \ \frac{v(\pi)}{l^{\eps}} \left( |\{k\le l^{\eps} \ : \ ({\cal I}_{ij}^{S}(\Pi))_k =\pi\}| \ - \ |\{k \le l^\eps : \Pi_k = \pi\}| \right) \nonumber \\
& = &  \sum \limits_{\pi \in \mathcal{P}_d} \ \frac{v(\pi)}{l^{\eps}} \left( |\{k \in S \ : \  {\cal I}_{ij}^{S}(\Pi_k) =\pi\}|  -  |\{k \in S : \Pi_k = \pi\}| \right).
 \label{partial_v}
\end{eqnarray}
Secondly, for every $j \in E$, 
\begin{eqnarray}
 \E_{\lambda,L^\eps,j}(|\{k\in S \ : \ \Pi_k = \pi \}|)   \ & =  &  \ \E_{\lambda,L^\eps,j}(\sum \limits_{k \in S} 1_{\{ \Pi_k = \pi \}})
 \ =  \   \ \sum \limits_{k \le l^\eps} 1_{\{ \Pi_k = \pi \}} \E_{\lambda,L^\eps,j}(1_{\{k \in S\}}). \nonumber 
 \end{eqnarray}
As $S$ is distributed as ${\cal F}_{j}^{ L^\eps, \lambda}$, we can use the fact that ${\mathbb P} \left(k\in {\cal F}_{j}^{ L^\eps, \lambda} \right) \ =   \ M_{ij} \frac{1}{ n^{\eps}_j}$ (see Remark \ref{rem:single-locus}), and then
 \begin{eqnarray}
 \E_{\lambda,L^\eps,j}(|\{k\in S \ : \ \Pi_k = \pi \}|)   & = &  \frac{1}{n^{\eps}_j}  | \{ k \le l^\eps, \ \Pi_k = \pi \} |  =   \frac{l^{\eps}}{n^{\eps}_j} X(\Pi)(\pi). \label{eq:e-sx} 
\end{eqnarray}
Furthermore, by applying \eqref{eq:e-sx} for every $\pi' \in \mathcal{I}_{ij}^{-1}(\pi)$ and then taking the sum over every such partitions, we get
$$
 \E_{\lambda,L^\eps,j}(|\{k\in  S \ : \ {\cal I}_{ij}(\Pi_k) = \pi\}|   \ =  \  \frac{l^{\eps}}{n^{\eps}_j}  X(\Pi)({\cal I}_{ij}^{-1}(\pi)) =  \frac{l^{\eps}}{n^{\eps}_j} \  {\cal I}_{ij}*  X(\Pi)(\pi).$$ 
This completes the proof of \eqref{eq:exp-partial}.  From this result, we deduce that
$$
{\mathbb G}^{\eps,L^\eps} g(\Pi) \ = \frac1{\eps} \sum_{i,j=1}^d M_{ij} \frac{1}{n^{\eps}_j}  ({\cal I}_{ij}*X(\Pi) - X(\Pi)) \ + \ b_{\infty} \sum_{i=1}^d \left({\cal S}_i*X(\Pi) - X(\Pi)\right).
$$
This completes the proof of Lemma \ref{lem:gen-first-order}.
\end{proof}

For every $L^\eps \in [0,1]^{l^\eps}$,  for every $v \in {\cal M}_d$, define
\begin{equation*}
M^{\eps, L^\eps,v}_t := \left<\xi_t^{\eps, L^\eps},v\right> \ - \  \int_0^t \left<^t G^\eps \xi_s^{\eps, L^\eps},v\right> ds, \ \  B^{\eps, L^{\eps, L^\eps},v}_t \ := \  \int_0^t \left<^t G^\eps \xi_s^{\eps, L^\eps},v\right> ds.
\end{equation*}
Since $\left<\xi_t^{\eps, L^\eps},v\right>$ is bounded, the previous result implies that $M^{\eps,L^\eps,v}$ is a martingale with respect to $({\cal H}^{L^\eps}_t)_{t\ge0}$, the filtration generated by  $(\Pi^{\eps,  L^\eps}(t); t\ge0)$. Further, the semi-martingale $\left<\xi_t^{\eps,L^\eps},v\right>$ admits the following decomposition:
$$
\left<\xi_t^{\eps, L^\eps},v\right>  \ = \  M^{\eps, L^\eps,v}_t \ + \ B^{\eps, L^\eps,v}_t. 
$$  
\begin{lemma}\label{lem:bracket}
For every $v \in {\cal M}_d$, for every $L^\eps \in [0,1]^{l^\eps}$,
$$
\left<M^{\eps, L^\eps,v}\right>_t \ = \   \int_0^t \ m^{\eps, L^\eps,v}(\Pi^{\eps,L^\eps}(s) ) ds 
$$
with  
$$
m^{\eps, L^\eps,v}(\Pi) = \frac1{\eps } \sum_{i,j= 1}^d \ M_{ij} \ \E_{\lambda,L^\eps,j} \left( \left<\partial_{ij}^S X(\Pi),v\right>^2\right)  \ +  \ b_{\infty} l^\eps \sum \limits_{i=1}^d \E_{l^\eps}\left(\left<\partial_i^K X(\Pi),v \right>^2\right) 
$$ 
and where  $\left<M^{\eps, L^\eps,v}\right>_t$ denotes the quadratic variation of $M^{\epsilon, L^\epsilon, v}$.
\end{lemma}
\begin{proof}
For $v\in {\cal M}_d$, define $h(\Pi)=\left<X(\Pi),v\right>^2$. Then, by   \eqref{eq:gener-x}
\begin{eqnarray*}
{\mathbb G}^{\eps, L^\eps} h(\Pi) & = & \frac{1}{\eps}  \sum \limits_{i,j=1}^d M_{ij}  \E_{\lambda,L^\eps,j} \left(\left<X({\cal I}_{ij}^S(\Pi),v)\right>^2\ -\ \left<X(\Pi),v\right>^2\right)   \nonumber \\
&& + \  b_{\infty} l^{\eps}  \sum \limits_{i=1}^d \E_{l^\eps}\left(\left<X({\cal S}_i^K(\Pi)),v\right>^2 - \left<X(\Pi),v\right>^2\right)  .
\end{eqnarray*}
Since
\begin{eqnarray*}
\left<X({\cal I}_{ij}^S(\Pi)),v\right>^2\ -\ \left<X(\Pi),v\right>^2 & = &  \left<X({\cal I}_{ij}^S(\Pi))-X(\Pi),v\right>^2 \ + \ 2 \left<X({\cal I}_{ij}^S(\Pi))-X(\Pi),v\right> \ \left<X(\Pi),v\right> \nonumber \\
\left<X({\cal S}_i^k(\Pi)),v\right>^2\ -\ \left<X(\Pi),v\right>^2 & = &  \left<X({\cal S}_i^k(\Pi))-X(\Pi),v\right>^2 \ + \ 2 \left<X({\cal S}_i^k(\Pi))-X(\Pi),v\right> \ \left<X(\Pi),v\right>,  \nonumber
\end{eqnarray*}
the previous identities yield
\begin{eqnarray*}
{\mathbb G}^{\eps, L^\eps} h(\Pi) \ &= & \  2 {\mathbb G}^{\eps, L^\eps} g(\Pi)\left<X(\Pi),v\right> \ + \\
&& \frac{1}{\eps}  \sum \limits_{i,j = 1}^d M_{ij} \E_{\lambda,L^\eps,j} \left(\left<\partial_{ij}^SX(\Pi),v\right>^2\right) 
 + \ b_{\infty} l^\eps \sum \limits_{i=1}^d \E_{l^\eps}\left(\left<\partial_i^K X(\Pi),v\right>^2\right) , \nonumber
\end{eqnarray*}
where $g(\Pi)= \left<X(\Pi),v\right>$. As a consequence
\beqnn
\left<X(\Pi^{\eps,L^\eps}(t)),v\right>^2 \ - \  2\int_0^t  \ {\mathbb G}^{\eps, L^\eps} g(\Pi^{\eps,L^\eps}(s))\left<X(\Pi^{\eps,L^\eps}(s)),v\right> ds  \nonumber \\
 - \int_0^t \ \frac{1}{\eps} \sum \limits_{i,j=1}^d M_{ij}  \E_{\lambda,L^\eps,j} \left(\left<\partial_{ij}^SX(\Pi^{\eps, L^\eps}(s)),v\right>^2\right) ds -
\    \int_0^t \ b_{\infty} l^\eps \sum \limits_{i=1}^d \E_{l^\eps}\left(\left<\partial_i^K X(\Pi^{\eps, L^\eps}(s)),v\right>^2\right) ds \nonumber
\eeqnn
is a martingale.
Further using It\^o's formula, the process 
$$
\left<X(\Pi^{\eps,L^\eps}(t)),v\right>^2 - 2\int_0^t  {\mathbb G}^{\eps, L^\eps} g(\Pi^{\eps,L^\eps}(s))\left<X(\Pi^{\eps,L^\eps}(s)),v\right> ds \ - \ \left<M^{\eps, L^{\eps},v}\right>_t
$$ 
is also a martingale. Combining the two previous results completes the proof of Lemma \ref{lem:bracket}. 
\end{proof}

\begin{prop}\label{prop:conv:m} 
$$
  \lim \limits_{\eps \to 0} \ \E( \sup_{\Pi \in ({\cal P}_d)^{l^\eps}}m^{\eps, {\cal L}^\eps,v}(\Pi)) =  \  0,$$
  where the expected value is taken with respect to the random variable $ {\cal L}^\eps$.
\end{prop}

\begin{prop}\label{prop:tightness} For $T>0$, 
the family of random variables $(\xi^\eps; \ {\eps>0})$ is tight in the weak topology $D([0,T], \R^{{Bell}_d})$.
\end{prop}
We postpone the proof of Propositions \ref{prop:conv:m} and \ref{prop:tightness} until Sections \ref{section:conv-m} and \ref{section:tightness} respectively. 

\begin{proof}[Proof of Theorem \ref{thm31} based on Proposition \ref{prop:conv:m} and \ref{prop:tightness}]
Since $(\xi^\eps; \ {\eps>0})$ is tight, we can always extract a subsequence 
converging in distribution (for the weak topology) to a limiting random measure process $\xi$.  We will now show that 
$\xi$ can only be the solution of the Kolmogorov equation alluded to in Theorem \ref{thm31}.
From \eqref{Gepsdef}, for every probability measure $m$ on  ${\cal P}_d$, for every $v\in {\cal M}_d$, 
\beqnn
\left| \ \left<^t G^\eps m, v\right>  \ \right| & \leq & (2 \sum \limits_{i,j=1}^d M_{ij} \ \frac{1}{\eps n^{\eps}_i} + 2 b_{\infty} d ) ||v||_\infty, \label{eq:bound-gen}
\eeqnn
where $|| v ||_{\infty} := \max_{\pi\in{\cal P}_N} v(\pi)$.
Since as $\eps\to0$,  $ n_i^\eps \eps \to  N_i $,  the term between parentheses also converges, and thus  
the RHS is uniformly bounded in $\eps$. Finally, the bounded convergence theorem implies that for every $v\in {\cal M}_d$,
\begin{equation*}
\E\left(\left<\xi_t,v\right> - \int_0^t \left<^tG \xi_s,v\right> ds \right)^2 \ = \ \lim_{\eps\to0} \  \E \left( \left (\left<\xi_t^{\eps},v\right> - \int_0^t \left<^t G^\eps \xi^{\eps}_s,v\right> ds \right)^2 \right ),
\end{equation*} 
where we used the fact that $^t G^\eps m \ \to \ ^t G m$ for every $m\in{\cal M}_d$ (where $G$ is defined as in Theorem \ref{thm31}).
On the other hand, since
\begin{eqnarray*}
\E\left(\left<\xi_t^{\eps},v\right> - \int_0^t \left<^t G^\eps \xi^{\eps}_s,v\right> ds \right)^2  \ &=& \E \left (\E \left( \left (\left<\xi_t^{\eps,  {\cal L}^\eps},v\right> - \int_0^t \left<^t G^\eps \xi^{\eps,  {\cal L}^\eps }_s,v\right> ds \right)^2 \lvert \ {\cal L}^\eps \right ) \right ) \nonumber \\
 &=& \ \E( \E(\left<M^{\eps,{\cal L}^\eps,v}\right>_t \lvert \ {\cal L}^\eps))\nonumber \\
&=& \E( \E ( \int_0^t \ m^{\eps, {\cal L}^\eps,v}(\Pi^{\eps,{\cal L}^\eps}(s) ) ds \ \lvert \ {\cal L}^\eps  )) \nonumber \\
&\le&  t  \  \E(\ \sup_{\pi \in ({\cal P}_d)^{l^\eps}}m^{\eps, {\cal L}^\eps,v}(\Pi)).   \nonumber \\
\end{eqnarray*}
Lemma \ref{lem:bracket} and
Proposition \ref{prop:conv:m} imply that 
$$
\E\left(\left<\xi_t,v\right> - \int_0^t \left<^tG \xi_s,v\right> ds \right)^2  \ = \ 0,
$$
which ends the proof of Theorem \ref{thm31}.
\end{proof}
\bigskip

\subsection{Proof of Proposition \ref{prop:conv:m}}
\label{section:conv-m}
Our first step in proving Proposition \ref{prop:conv:m} is to prove the following result.

\begin{lemma}\label{lemma:ARG}$\forall \ j \ \in  E$, $\forall \lambda >0$, 
\begin{eqnarray*} 
\lim \limits_{\eps \to 0 } \  \frac1{(l^{\eps})^2 \eps} \  \E \left ( \E_{\lambda,{\cal L}^\eps,j} ( \sum  \limits_{k=1}^{l^{\eps}} \sum \limits_{k'=1}^{l^\eps}   1_{k \in S}  1_{k' \in S} \ \lvert \ {\cal L}^\eps ) \right ) \ = \ 0. \end{eqnarray*}
\end{lemma}

Before turning to the proof of this result, we recall the definition of the ancestral recombination graph (ARG) (see also \cite{hudson}, \cite{ARG1}, \cite{ARG2}) for the case of two loci.
Fix ${L}^\eps = \{ \ell_1, \dots  \ell_{l^\eps}\}$ the positions of the loci in the chromosome, $\lambda$ the recombination rate, and choose two loci $k$ and ${k'}$ among the $l^\eps$ loci. In order to compute the probability that for both loci the allele from the migrant is fixed in the host population -- $\mathbb{P}(k,k' \in {\cal F}_{j}^{L^\eps, \lambda})$ --  we  follow backwards in time the genealogy of the corresponding alleles carried by a reference individual in the present population, assuming that a migration event occurred in the past (sufficiently many generations ago, so that we can assume that the present population is homogeneous). 
\begin{figure}[h]
\begin{center}
\subfigure[Moran model with recombination]{\includegraphics[height=8cm]{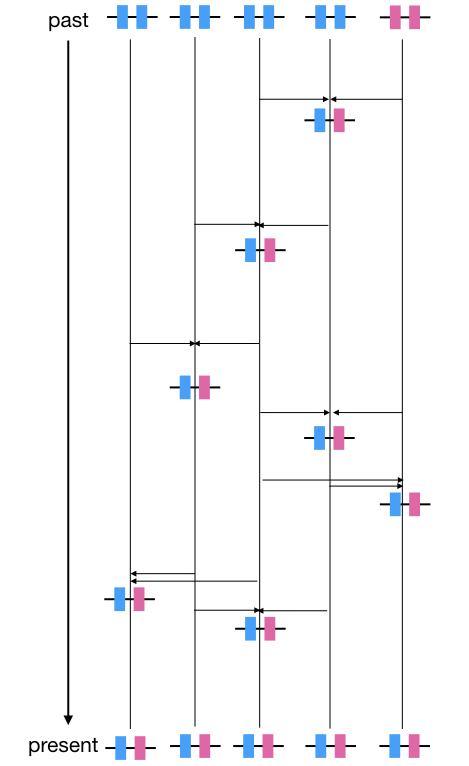}} \qquad
\subfigure[ARG with 2 loci]{\includegraphics[height=8cm]{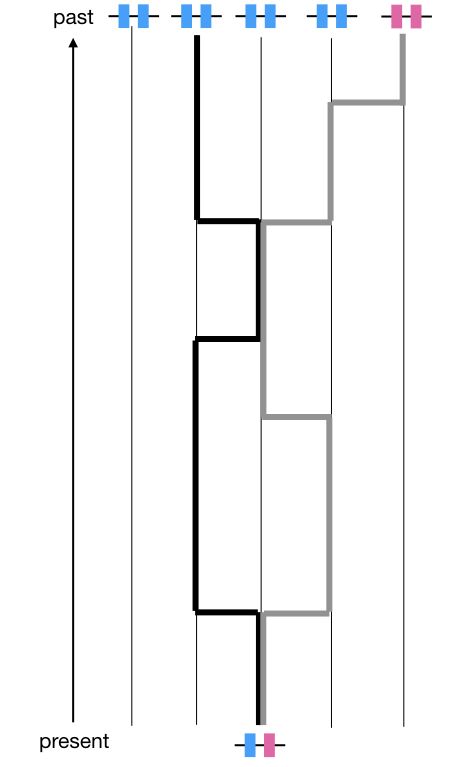}}\qquad
\caption{Realisation of a Moran model with recombination and the corresponding ARG. The population size is equal to 5 and $l^\eps=2$. The colors of the loci represent their origin: loci from the resident subpopulation are represented in blue whereas loci from the migrant are represented in pink.   In (a) the origins of the arrows indicate the parents, and the tips of the arrows point to their offspring. Time goes from top to bottom as indicated by the arrow on the left. In (b) the black (resp. grey) line corresponds to the ancestral lineages of the first (resp. second) locus in the chromosome of an individual sampled in the extant population. Time goes backwards, from  bottom to  top. The first locus has been inherited from the resident subpopulation whereas the second locus has been inherited from the migrant (i.e. $1 \ne S, \ 2 \in S$). } 
\label{fig2}
\end{center}
\end{figure}
More precisely, at locus $k$, we consider the ancestral lineage of a reference individual (chosen uniformly at random) in the extant population. We envision this lineage as a particle moving in $\{1,\dots,n^{\eps}_j\}$: time $t=0$ corresponds to the present, and the position of the particle at time $t$ -- denoted by $A^{L^\eps, \lambda, j}_k(t)$ -- identifies the ancestor of locus $k$, $t$ units of time in the past (i.e., at locus $k$, the reference individual  in the extant population inherits its genetic material from individual $A^{L^\eps, \lambda, j}_k(t)$ at time $-t$) (see Figure \ref{fig2}). 

The recombination rate between the two loci, $k$ and $k'$, $r^{L^\eps, \lambda}_{k,k'}$ corresponds to the probability that there is at least one Poisson point between $ \ell_{k}$ and $ \ell_{k'}$ and the two fragments are inherited from different parents and is given by 
\be\label{rate-of-recombination}
r^{L^\eps, \lambda}_{k,k'} := \frac12 \left (1- \exp(-\lambda | \ell_k -  \ell_{k'}|) \right ).
\ee
  
 ${\bf A}^{L^\eps, \lambda, j}=(A^{L^\eps, \lambda, j}_k, \   A^{L^\eps, \lambda, j}_{k'})$ defines a $2$-dimensional stochastic process on $\{1,\dots,n_j^\eps\}$. At time $0$, the two particles have the same position (they coincide at a randomly chosen individual as in Figure \ref{fig2}) and then evolve according to the following dynamics:
 \begin{itemize}
 \item When both particles are occupying the same location $z$, the group splits into two at rate $r^{L^\eps, \lambda}_{k,k'}$ (see \eqref{rate-of-recombination}). Forward in time this corresponds to a reproduction event where $z$ is replaced by the offspring of $x$  and $y$. Each individual $x$ reproduces at rate 1 (chooses a random partner $y$), and with probability $1/n_j^\eps$ his offspring replaces individual $z$. There are $n_j^\eps$ possible choices for  $x$. Following   \eqref{rate-of-recombination}, the probability that both loci are inherited from different parents is $ r^{L^\eps, \lambda}_{k,k'}$,  so the rate of fragmentation for loci $k,k'$ is given by $n_j^\eps \ .\  \frac1{n_j^\eps} \ . \  r^{L^\eps, \lambda}_{k,k'}$. 
 \item When the two particles are occupying different positions, they jump to the same position at rate $2/n_j^\eps$. 
 Forwards in time, this corresponds to a reproduction event where the individual located at $A^{L^\eps, \lambda, j}_k$ (resp. $A^{L^\eps, \lambda, j}_{k'}$) replaces the one at $A^{L^\eps, \lambda, j}_{k'}$ (resp. $A^{L^\eps, \lambda, j}_k$), and the offspring inherits the allele at locus $k'$ (resp. $k$) from this parent.
 A reproduction event where the individual located at $A^{L^\eps, \lambda, j}_k$ (resp. $A^{L^\eps, \lambda, j}_{k'}$) replaces the one at $A^{L^\eps, \lambda, j}_{k'}$ (resp. $A^{L^\eps, \lambda, j}_k$) occurs at rate $2/n_j^\eps$ (as the individual at $A^{L^\eps, \lambda, j}_k$ --resp. $A^{L^\eps, \lambda, j}_{k'}$-- can be the mother or the father); and the probability that  the offspring inherits the locus $k'$ (resp. $k$) from this parent is $1/2$. The total rate of coalescence is $ 2 \ . \ \frac2{n_j^\eps} \ . \ \frac12$. 
 \end{itemize}

Since we assume that the migration event occurred far back in the past, the following duality relation holds:
\be\label{eq:duality}
{\mathbb P}(k,k' \in {\cal F}_{j}^{L^\eps, \lambda}) \ = \ \lim_{t\to\infty} \ {\mathbb P}\left( \ A^{L^\eps, \lambda, j}_k(t) = A^{L^\eps, \lambda, j}_{k'}(t) = 1  \right).
\ee
 In other words, assuming that the migrant is labelled 1, the set on the RHS corresponds to the set of loci inheriting their genetic material from the migrant.
  
 \begin{proof}[Proof of Lemma \ref{lemma:ARG}]
Define $(Y^{L^\eps, \lambda, j}(t):=1_{A^{L^\eps, \lambda, j}_k(t)=A^{L^\eps, \lambda, j}_{k'}(t)}; t\geq0)$. It is easy to see from the previous description of the dynamics that $Y$ is a Markov chain on $\{0,1\}$ with the following transition rates:
$$
 q_{1,0} =  r^{L^\eps, \lambda}_{k,k'}, \ \ q_{0,1} = \frac2{n_j^\eps}
 $$
 and further
 \begin{itemize}
\item conditional on $Y^{L^\eps, \lambda, j}(t)=1$,  the two lineages $(A^{L^\eps, \lambda, j}_k(t),A^{L^\eps, \lambda, j}_{k'}(t))$ occupy a common position that is distributed  as a uniform random variable on $\{1,\dots,n_j^\eps\}$.
\item conditional on $Y^{L^\eps, \lambda, j}(t)=0$, $(A^{L^\eps, \lambda, j}_k(t),A^{L^\eps, \lambda, j}_{k'}(t))$ are distinct and are distributed as a two uniformly sampled random variables (without replacement)  on $\{1,\dots,n_j^\eps\}$.
\end{itemize}

We have:
\begin{equation*}
 \mathbb{P} (A^{L^\eps, \lambda, j}_k(t) =  A^{L^\eps, \lambda, j}_{k'}(t) = 1)  =  \mathbb{P} \left(Y^{L^\eps, \lambda, j}(t)=1\right) \frac{1}{n_j^\eps}\ . 
\end{equation*}
Furthermore, it is straightforward to show that 
$$ \lim_{t\to\infty} \mathbb{P} \left(Y^{L^\eps, \lambda, j}(t)=1\right) \ = \ \frac2{n_j^\eps}\frac1{ r^{L^\eps, \lambda}_{k,k'} + \frac2{n_j^\eps}}. $$ 
From (\ref{eq:duality}), we get that,
\begin{eqnarray*}
&&   \  \E_{\lambda,L^\eps,j} ( \sum  \limits_{k, k' \in \{1, \dots, l^\eps\}}  1_{k \in S}  1_{k' \in  S}) \  = \  \lim \limits_{t \to \infty}  \sum  \limits_{k, k' \in \{1, \dots, l^\eps\}}    \mathbb{P} (A^{L^\eps, \lambda, j}_k(t) = A^{L^\eps, \lambda, j}_{k'}(t)=1) \nonumber \\ 
& = & \frac2{(n^{\eps}_j)^2}  \sum_{k, k' \in \{ 1,\dots, l^\eps \}} \frac{1}{\frac12(1-e^{-\lambda| \ell_k- \ell_{k'}|}) + \frac2{n^{\eps}_j}}.      \nonumber
\end{eqnarray*}
One can then easily check that, $\exists \  \alpha > 0$ such that, for every ${  L}^\eps = \{  \ell_1, \dots  \ell_{l^\eps}\}$,
\begin{equation}
   \  \E_{\lambda,L^\eps,j} ( \sum  \limits_{k, k' \in \{1, \dots, l^\eps\}}  1_{k \in S}  1_{k' \in  S}) \ \leq \    \frac2{(n^{\eps}_j)^2}  \sum_{k, k' \in \{1, \dots, l^\eps\}} \frac{1}{\alpha | \ell_k- \ell_{k'}| + \frac2{n^{\eps}_j}}. 
 \label{majoration-alpha}
\end{equation}
Thus,
\begin{eqnarray*}
\E \left ( \E_{\lambda,{\cal L}^\eps, j} ( \sum  \limits_{k=1}^{l^{\eps}} \sum \limits_{k'=1}^{l^\eps}   1_{k \in S}  1_{k' \in S}  \ \lvert \ {\cal L}^{\eps})\right ) & \le &   \frac2{(n^{\eps}_j)^2}  \int_{[0,1]^{l^\eps}} dx_1, \dots dx_{l^\eps} \sum_{k, k' \in \{1, \dots, l^\eps\}} \frac{1}{\alpha |x_k-x_{k'}| + \frac2{n^{\eps}_j}} \\ \nonumber 
& \le &  \frac2{(n^{\eps}_j)^2}  \sum_{k, k' \in \{1, \dots, l^\eps\}} \int_{[0,1]^2}\frac{dx_k dx_{k'}}{\alpha |x_k-x_{k'}| + \frac2{n^{\eps}_j}}.   \nonumber
\end{eqnarray*}
In addition, using the fact that $n_j^\eps  = [N_{j}/\eps] $,
\begin{eqnarray*}
&& \frac{1}{(l^\eps)^2 \eps}  \E \left ( \E_{\lambda,{\cal L}^\eps,j} ( \sum  \limits_{k=1}^{l^{\eps}} \sum \limits_{k'=1}^{l^\eps}   1_{k \in S}  1_{k' \in S} \ \lvert \ {\cal L}^{\eps} )\right ) \nonumber \\
&\le & \   \frac{2 \eps }{(N_j-\eps)^2} \int_{[0,1]^2}  \frac{dt ds}{\alpha|t-s| +2\eps/N_j}    \nonumber \\
& =&  \frac{4 \eps }{(N_j-\eps)^2} \int_0^1 ds \int_o^s  \frac{dt}{\alpha|t-s| +2\eps/N_j}  \  \nonumber \\
&= &  \frac{4 \eps }{\alpha (N_j-\eps)^2} \int_0^1   \log \left (\frac{\alpha N_j}{2\eps} s + 1 \right )ds  \    \nonumber \\
& =  &\frac{4\eps}{ \alpha (N_j-\eps)^2}    \left ((1+  \frac{ 2\eps}{\alpha N_j}) \log \left (\frac{\alpha N_j}{2\eps }  + 1 \right) -1 \    \right ) \   \nonumber \\
 &\underset{\eps \to 0}{\longrightarrow}& 0. \nonumber 
\end{eqnarray*}
 \end{proof}

We are now ready to prove Proposition \ref{prop:conv:m}.

\begin{proof}[Proof of Proposition \ref{prop:conv:m}]
Using the definition given in Lemma \ref{lem:bracket},
$$
m^{\eps, L^\eps,v}(\Pi) = \frac1{\eps } \sum_{i,j= 1}^d \ M_{ij} \ \E_{\lambda,L^\eps,j} \left( \left<\partial_{ij}^S X(\Pi),v\right>^2\right)  \ +   \ b_{\infty} l^\eps \sum \limits_{i=1}^d \E_{l^\eps}\left(\left<\partial_i^K X(\Pi),v \right>^2\right).
$$

To bound the second term on the RHS, we note that, by definition, ${\cal S}_i^k(\Pi)$ and $\Pi$ only differ in one component, so from the definition of $\partial_i^K X(\Pi)$ (see \eqref{defdronde}), it is not hard to see that 
  \begin{equation*}
   \left<\partial_i^K X(\Pi),v \right> ^2 \ \leq \  \frac{4}{(l^\eps)^2} || v ||_{\infty}^2.
  \end{equation*}
  It follows that,
\begin{equation}
 {b_{\infty} l^\eps} \ \E_{l^\eps}(\left<\partial_i^K X(\Pi),v\right>^2) \leq \frac{4  b_{\infty}}{l^\eps } || v ||_{\infty}^2.  
 \label{majoration-rhs2} \end{equation}
 Since  $l^{\eps} \to \infty$ as $\eps \to 0$, this term converges and can be bounded from above, uniformly in $\Pi$ and $\eps\in(0,1)$. Note that this bound does not depend on the choice of ${L}^\eps$.

For the second term on the RHS, we simply note that expanding $ \frac1{\eps }  \E_{\lambda,L^\eps,j} \left(\left<\partial_{ij}^S X(\Pi),v\right>^2\right)$ (see   \eqref{partial_v}), yields a sum of four terms  that can be upper bounded by
$$
\frac{ || v ||_{\infty}^2}{(l^\eps)^2 \eps   }\E_{\lambda, L^\eps, j} (|k \in S, \ \Pi_k \in p_1| \ |k \in S, \ \Pi_k \in p_2|),
$$ where $p_1$ and $p_2$ are alternatively replaced by $\{\pi\}, {\cal I}_{ij}^{-1}(\pi)$ with $\pi\in{\cal P}_d$. Finally, $\forall L^\eps \in [0,1]^{l^\eps}$,
\begin{eqnarray}
&&\frac{ \E_{\lambda,L^\eps,j} (|k \in S, \ \Pi_k \in p_1| \ |k \in S, \ \Pi_k \in p_2|)}{(l^\eps)^2 \eps }\nonumber \\
  &=& \frac1{(l^\eps)^2 \eps}  \E_{\lambda,L^\eps,j}(\sum \limits_{k=1}^{l^\eps} 1_{\Pi_k \in p_1} 1_{k \in S}\sum \limits_{k'=1}^{l^\eps} 1_{\Pi_{k'} \in p_2} 1_{k' \in S}) \nonumber \\
&=& \frac1{(l^\eps)^2 \eps } \sum  \limits_{k=1}^{l^\eps} \sum \limits_{k'=1}^{l^\eps}  1_{\Pi_k \in p_1} 1_{\Pi_{k'}\in p_2}  \E_{\lambda,L^\eps,j} (1_{k \in S}  1_{k' \in S}) \nonumber \\
&\le& \frac1{(l^\eps)^2 \eps }   \E_{\lambda,L^\eps,j} ( \sum  \limits_{k=1}^{l^\eps} \sum \limits_{k'=1}^{l^\eps}   1_{k \in S}  1_{k' \in S}) \nonumber \\
\label{majoration-rhs}
\end{eqnarray}
randomising the positions of the loci and using Lemma \ref{lemma:ARG} the term on the RHS also converges and can also be bounded from above, which completes the proof.
\end{proof}

\begin{rem}[Magnitude of the stochastic fluctuations]
Lemma \ref{lem:gen-first-order} and the proof Proposition \ref{prop:conv:m} entail that:
$$\forall v\in {\cal M}_d,  \ \ \E(\left<M^{\eps,{\cal L}_\eps,v}\right>_t) \ \le \eps \log(1/\eps) C + \frac1{l^\eps} C'$$
where $C, C'$ are constants.
This suggests that the order of magnitude of the fluctuations should be of the order of  $\max(\sqrt{\eps \log(1/\eps)}, \sqrt(1/{l^\eps}))$.

In \cite{Yama}, the authors proposed a diffusion approximation (only for the case of two subpopulations). Their approximation is based on the simplifying hypothesis that loci are fixed independently on each other -- the number of fixed loci (after each migration event) follows a binomial distribution--, and the hypothesis that the number of loci $l$ is s.t. $l >> \frac1{\eps}$. They found that the magnitude of the stochastic fluctuations was $\sqrt{\eps}$. 

In summary, the previous heuristics suggest that taking into account correlations between loci increases the magnitude of the stochastic fluctuations. \end{rem}

\subsection{Tightness: Proof of Proposition \ref{prop:tightness}} \label{section:tightness}
 We follow closely \cite{fournier2004}. 
It is sufficient to prove that for every $v\in{\cal M}_d$, the projected process $\left( \left<\xi^\eps,v\right>; \ \eps>0 \right)$ is tight. To this end, we use Aldous criterium (see \cite{aldous1989}). In the following, we define
\begin{equation*}
M^{\eps,v}_t := \left<\xi_t^\eps,v\right> \ - \  \int_0^t \left<^t G^\eps \xi_s^\eps,v\right> ds, \ \  B^{\eps,v}_t \ := \  \int_0^t \left<^t G^\eps \xi_s^\eps,v\right> ds.
\end{equation*}
We first note that
$$
\sup_{t\in[0,T]} |\left<\xi^\eps_t,v\right>| \ \leq \ || v ||_{\infty}, 
$$ 
which implies that that for every deterministic $t\in[0,T]$, the sequence of random variables  $(\left<\xi^\eps_t,v\right>; \eps>0)$
is tight. Thus, the first part of Aldous criterium is sastified. Next,
fix $\delta>0$, and take two stopping times $\tau^\eps$ and $\sigma^\eps$ with respect to $({\cal H}^{\eps}_t)_{t\ge 0}$ the filtration generated by $(\Pi^{\eps, {\cal L}^\eps}_t, t\ge0)$, such that $0\leq \tau^\eps\leq \sigma^\eps \leq \tau^\eps +\delta\leq T$. Since 
$
\left<\xi_t^{\eps},v\right>  \ = \  M^{\eps,v}_t \ + \ B^{\eps,v}_t,
$
it is enough to show that the quantities
\[\E\left(\left|M^{\eps, v}_{\sigma^\eps} - M^{\eps, v}_{\tau^\eps}\right| \right) \ \ \mbox{and} \ \ \E \left( \left| B_{\sigma^\eps}^{\eps,  v} \ - \ B_{\tau^\eps}^{\eps, v} \right| \right ) \]
are bounded from above by two functions in $\delta$ (uniformly in the choice of $\tau^\eps,\sigma^\eps$ and $\eps$) going to $0$ as $\delta$ go to $0$.

The rest of the proof is dedicated to proving those two inequalities. We start with the martingale part.
First, 
\begin{eqnarray*}
\E\left(\left|M^{\eps, v}_{\sigma^\eps} - M^{\eps, v}_{\tau^\eps}\right| \right)^2 
& \leq & \E \left( \left(M^{\eps, v}_{\sigma^\eps} - M^{\eps,  v}_{\tau^\eps}\right)^2 \right) 
\end{eqnarray*} 
Recall that $\forall \ L^\eps \in [0,1]^{l^\eps}$, $M^{\eps, L^\eps, v}$ is a martingale. Thus, $M^{\eps, v}$ is a martingale with respect to   $({\cal G}_t^{\eps})_{t\ge 0} = ({\cal H}^{\eps}_t)_{t\ge 0} \vee \sigma({\cal L}^\eps)$, where $({\cal H}^{\eps}_t)_{t\ge 0}$ is the filtration generated by $(\Pi^{\eps, {\cal L}^\eps}_t)$. As $({\cal H}^\eps_t)_{t\ge0} \subset ({\cal G}^\eps_t)_{t \ge0}$, $\tau^\eps$ and $\sigma^\eps$ are also stopping times for the filtration  $({\cal G}^{\eps}_t)_{t\ge 0}$, so that
\begin{eqnarray*}
\E\left(\left|M^{\eps, v}_{\sigma^\eps} - M^{\eps, v}_{\tau^\eps}\right| \right)^2
&  \leq &  \E \left( \E (\left(M^{\eps, {\cal L}^\eps, v}_{\sigma^\eps} - M^{\eps, {\cal L}^\eps, v}_{\tau^\eps}\right)^2 \ | \ {\cal L}^\eps )\right) \\
& \leq  & \E \left( \E \left(\left<M^{\eps,{\cal L}^\eps, v}\right>_{\sigma^\eps} - \left<M^{\eps,{\cal L}^\eps, v}\right>_{\tau^\eps} \lvert \ {\cal L}^\eps\right)\right) \\
& = & \E \left(  \int_{\sigma^\eps}^{\tau^\eps} \ m^{\eps, {\cal L}^\eps,v}(\Pi^\eps(s) ) ds \right ).
\end{eqnarray*}
where $m^{\eps, L^\eps,v}(\Pi)$  was defined in Lemma \ref{lem:bracket} and where the second line follows from the fact that 
$\tau^\eps$ and $\sigma^\eps$ are stopping times for the filtration  $({\cal G}^{\eps}_t)_{t\ge 0}$.

If there exists $C_1$ such that \begin{equation}\sup_{L^\eps \in [0,1]^{l^\eps}}\sup_{\Pi \in ({\cal P}_d)^{l^\eps}} m^{\eps, L^\eps,v}(\Pi) \le C_1, \label{majorationm} \end{equation}
then, 
$$
 \E\left(\left|M^{\eps, v}_{\sigma^\eps} - M^{\eps, v}_{\tau^\eps}\right| \right)  \ \leq \ \sqrt{C_1} \sqrt{\delta},
$$
thus showing the desired inequality for the martingale part $M^{\eps, v}$.  To prove   \eqref{majorationm}, we recall the definition of $m^{\eps, L^\eps,v}(\Pi)$,
$$ m^{\eps, L^\eps,v}(\Pi) =  \frac1{\eps } \sum_{i,j= 1}^d \ M_{ij} \ \E_{\lambda,L^\eps,j} \left( \left<\partial_{ij}^S X(\Pi),v\right>^2\right)  \ +   \ b_{\infty}l^\eps \sum \limits_{i=1}^d \E_{l^\eps}\left(\left<\partial_i^K X(\Pi),v \right>^2\right).$$
The second term on the RHS can be bounded as in the proof of Proposition \ref{prop:conv:m}  (see \eqref{majoration-rhs2}). For the first term on the RHS, we use the bound given by   \eqref{majoration-rhs}. We only need to prove that   $\frac1{(l^{\eps})^2 \eps} \  \E_{\lambda,L^\eps,j} ( \sum  \limits_{k=1}^{l^{\eps}} \sum \limits_{k'=1}^{l^\eps}   1_{k \in S}  1_{k' \in S})$ is bounded.
Using   \eqref{majoration-alpha}, 
\begin{eqnarray*}
\frac1{(l^{\eps})^2 \eps} \  \E_{\lambda,L^\eps,j} ( \sum  \limits_{k=1}^{l^{\eps}} \sum \limits_{k'=1}^{l^\eps}   1_{k \in S}  1_{k' \in S}) \le \frac1{(l^{\eps})^2 \eps} \  \frac{(l^\eps)^2}{n_j^\eps} \underset{\eps \to 0}{\longrightarrow} N_j,
\end{eqnarray*}
so   \eqref{majorationm}  is proved.

\bigskip

We now turn to the drift part. First, for every $L^\eps \in [0,1]^{l^\eps}$,  
\begin{eqnarray*}
\left| B_{\sigma^\eps}^{\eps, L^\eps, v} \ - \ B_{\tau^\eps}^{\eps, L^\eps, v} \right| & \leq  & \int_{\tau^\eps}^{\sigma^\eps} | \left<^t G^\eps X(\Pi^{\eps,L^\eps}(s)),v\right> | ds.
\end{eqnarray*}
We already showed in \eqref{eq:bound-gen}, that the integrand on the RHS is uniformly bounded in $\eps$. Thus, there exists $C_2$ such that, for every $L^\eps \in [0,1]^{l^\eps}$: 
\begin{eqnarray*}
\left| B_{\sigma^\eps}^{\eps, L^\eps, v} \ - \ B_{\tau^\eps}^{\eps, L^\eps, v} \right| & \leq  & \delta C_2.
\end{eqnarray*}
So, 
\begin{eqnarray*}
\E \left( \left| B_{\sigma^\eps}^{\eps,  v} \ - \ B_{\tau^\eps}^{\eps, v} \right| \right ) &=& \E\left ( \E \left ((| B_{\sigma^\eps}^{\eps, {\cal L}^\eps, v} \ - \ B_{\tau^\eps}^{\eps, {\cal L}^\eps, v} |) \ \lvert \ {\cal L}^\eps   \right) \right)  \leq   \delta C_2.
\end{eqnarray*}
which is the desired inequality.
This completes the proof of Proposition \ref{prop:tightness}.
 
 \begin{rem}
Notice that the tightness (and the convergence) does not depend on the recombination rate. However, for small values of $\lambda$, or if $L^\eps$ is such that the positions of the loci are all very close to each other, correlations between loci are very high. This means that, when a migration event takes place, either no locus will be fixed (with high probability), or almost all loci from the migrant will be fixed. Therefore, if we let  $\lambda \to 0$ be the process of the genetics distances converges to a process that increases continuously (due to mutation) and has negative jumps (due to migration events).  See Figure \ref{jumps2}  for a numerical simulation. 
\end{rem}

\begin{figure}[h]  
\begin{center}
\includegraphics[width=10cm]{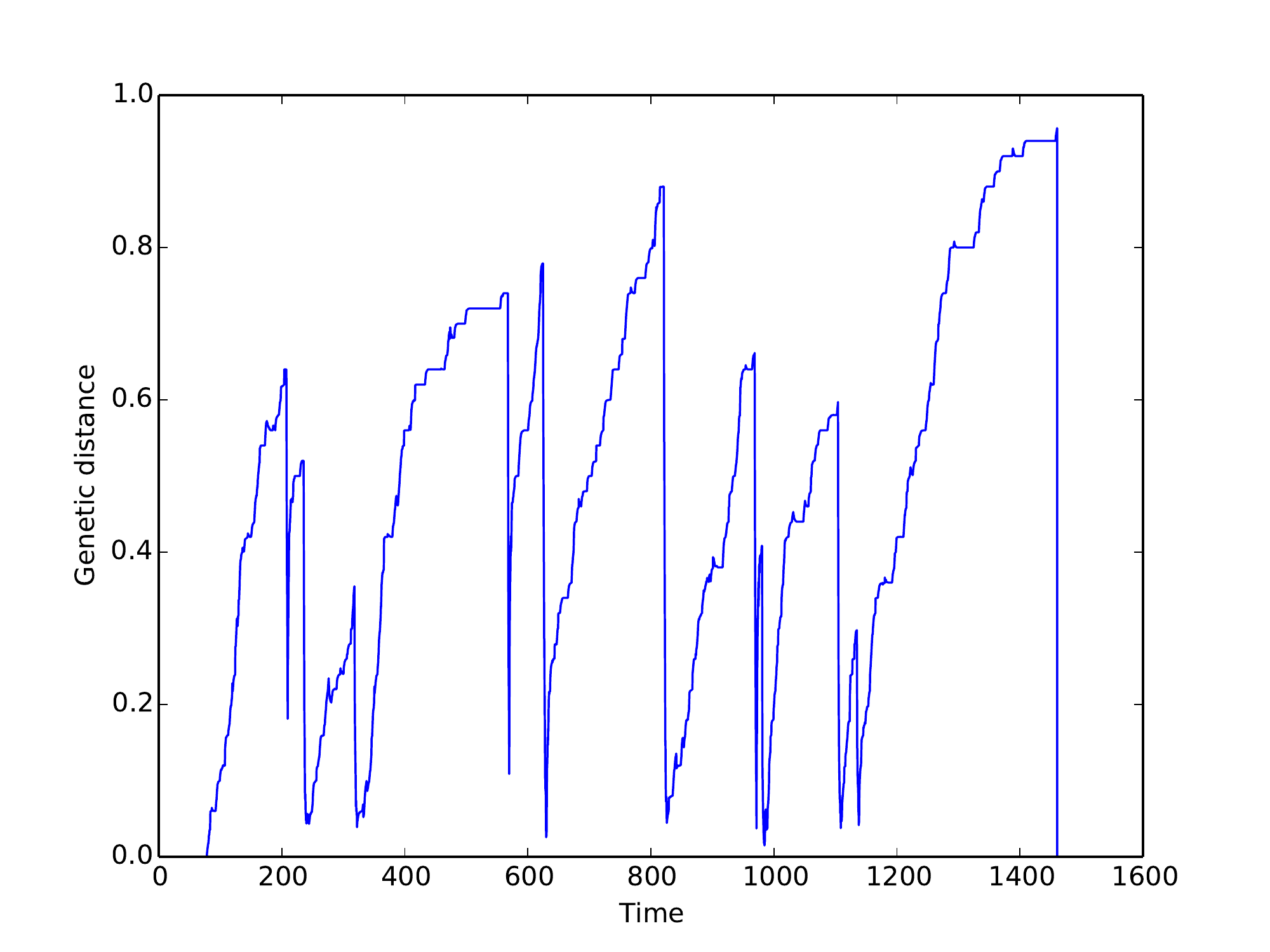}  
\caption{Simulation of the individual based model, for $d = 2$, $N_1 = N_2 = 1$, $\eps$ = 0.01, $\gamma = 0.005$, $l^\eps = 100$, $\lambda = 0.5$. With this set of parameters, Theorem \ref{thm-intro}, predicts that the genetic distance at equilibrium should be $0.5$. In this simulation, the mean genetic distance is $0.5$.}
\label{jumps2} 
\end{center}   
\end{figure}

\section{Proof of Theorem \ref{thm-intro} and more}
\label{sect:33}
In this section we state and prove a stronger version of Theorem \ref{thm-intro} (see Theorem \ref{teo-generak} below).
As in Theorem  \ref{thm-intro}, we consider, for each pair of subpopulations $i, j \in E$,  $S^i$ and $S^j$, two independent random walks on $E$ starting respectively from $i$ and $j$ and whose transition rate from $k$ to $p$ is equal to $\tilde M_{kp}$.
We have the following generalization of of Theorem \ref{thm-intro}.

\begin{theorem}\label{teo-generak}
Assume  that
\begin{itemize}
\item At time $0$, in the IBM, subpopulations are homogeneous and that, the genetic partition measure of the population (in the associated PBM) is given by $\xi^\eps_0$, a deterministic probability measure in  ${\cal P}_d$. 
\item There exists a probability measure $P^0 \in {\cal M}_d$ such that the following convergence holds: 
\begin{equation}\label{initial}
\xi^\eps_0 \underset{\eps \to 0}{\longrightarrow} P^0.
\end{equation}
\end{itemize}
For every $t\geq0$, define\
\begin{equation*}
D_t(i,j) \ := \ 1 \ - \  \int_0^t e^{-2b_\infty s} \Pp(\tau_{ij} \in ds) \ - \ \int_{\pi} e^{-2b_\infty t} \Pp\left(\tau_{ij} > t, S^i(t) \sim_{\pi} S^j(t)\right) P^0(d\pi)  \
\end{equation*}
where 
$\tau_{ij}=\inf\{t \geq0 \ : \ S^i(t) = S^j(t)\}.$
Then,
\begin{equation*}
\lim_{\eps \to 0} \ \lim_{\gamma \to 0} (d^{\gamma, \eps}_{t/(\gamma \eps)}(i,j), \ t \ge 0) \ = \ (D_t(i,j), \ t\ge 0)\  \textrm{ in the sense of finite dimensional distributions.}
\end{equation*}
In particular, 
\begin{equation*}
\lim_{t \to \infty} D_t(i,j) \ = \ 1 - \E(e^{-2 b_\infty \tau_{ij}}). 
\end{equation*}
\label{thm-main2}
\end{theorem}

\begin{proof} We start by proving that, 
\begin{equation}
(d_t^\eps(i,j); t \ge 0) \underset{\eps \to 0}{\Longrightarrow} (D_t(i,j); t\ge 0) \ \textrm{  in the weak topology,}
\label{convergence:eps}
\end{equation}
where $(D_t(i,j); t\ge 0)$ is the deterministic process defined in Theorem \ref{thm-intro}.

From equation \eqref{eq:pmgp-to-distance} and Theorem \ref{thm31} we get that $\forall i,j \in E$,  $(d_t^{\eps}(i,j);  \ t\ge 0 )$ converges in distribution in the weak topology to $(1-P_t( \pi \in {\cal P}_d, i \sim_{\pi} j); \ t\ge 0 )$. It remains to show that this expression is identical to the one provided in Theorem \ref{thm-intro}. This is done in a standard way by using the graphical representation associated to the one-locus Moran model whose generator is specified by $G$ (defined in   \eqref{def:G-inf}).  It is well known that such a Moran model is encoded by a graphical representation that is generated by a sequence of independent Poisson Point Processes as follows: 
\begin{itemize}
\item $B^i$,  with intensity measure $b_\infty dt$, that corresponds to mutation events at site $i$. If $b^i \in B^i$, at  $(i, b^i)$ we draw a $\star$ in the graphical representation (Figure \ref{fig1}(a)). 
\item $T^{i,j}$,  with intensity measure  $\tilde{M}_{ji}dt$,  that corresponds to reproduction events, where $j$ is replaced by $i$. If $t^{i,j} \in T^{i,j}$, we draw an arrow from $(i,t^{i,j} )$ to  $(j, t^{i,j} )$ in the graphical representation to indicate that lineage $j$ inherits the type of lineage $i$.
\end{itemize}

\begin{figure}[h]
\begin{center}
\subfigure[Graphical representation of the Moran model]{\includegraphics[width=4cm]{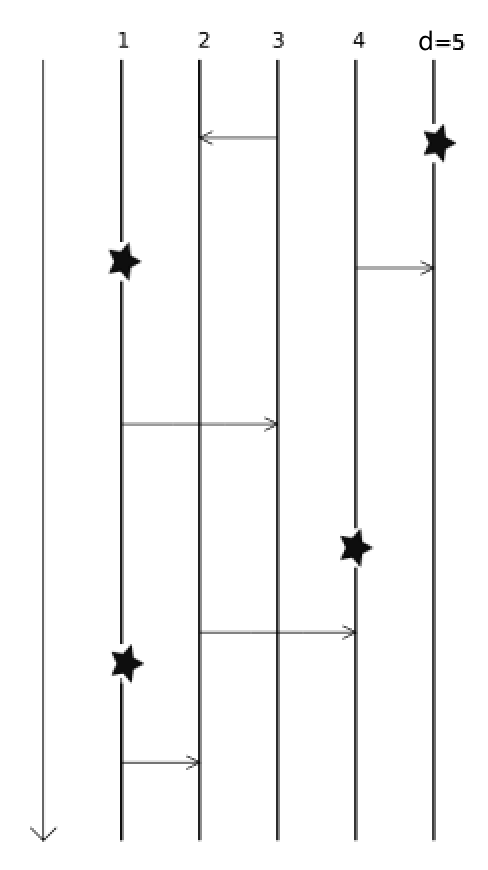}} \qquad
\subfigure[Genetic partitioning process]{\includegraphics[width=4cm]{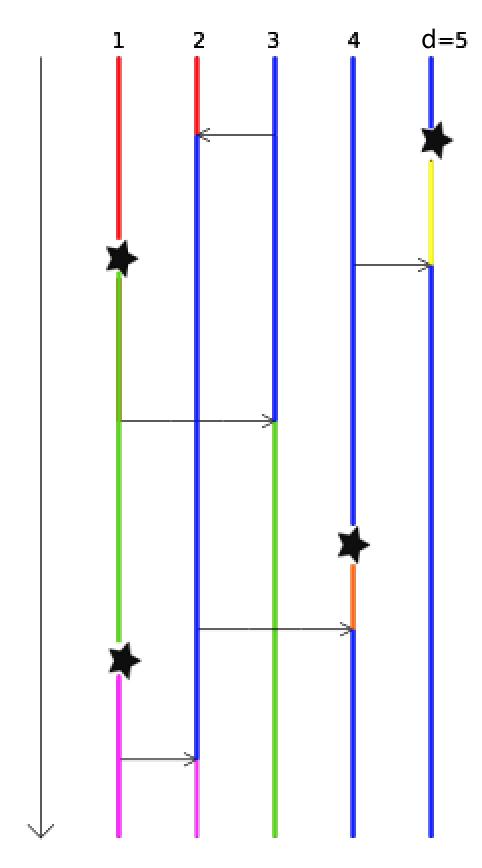}}\qquad
\subfigure[Ancestral lineages]{\includegraphics[width=4cm]{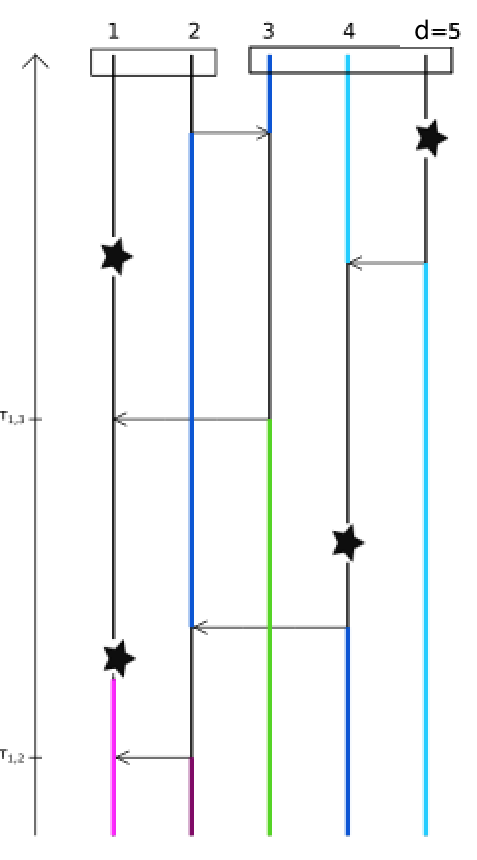}}
\caption{Realisation of the genetic partitioning model and its dual. In Figure (b), colours indicate genetic types (that induce the partitions). In Figure (c), colours represent the ancestral lineages.}
\label{fig1}
\end{center}
\end{figure}

We now give a characterisation of the dual process starting at $t$.  We define $(S^1_t, S^2_t,\dots, S^d_t)$ a sequence of piecewise continuous functions $[0, t] \to E$, where  $\forall  i \in E, \ S^i_t$ represents the ancestral lineage of individual $i$ (sampled at time $t$). $S^i_t(t) = i$ and as time proceeds backwards, each time $S^i_t$ encounters the tip of an arrow it jumps to the origin of the arrow. It is not hard to see that $S^i_t$ is distributed as a random walk started at $i$ and with transition rates from $k$ to $p$ equal to $\tilde{M}_{kp}$ and  that $(S^1_t, S^2_t,\dots, S^d_t)$ are distributed as coalescing random walks running backwards in time, i.e. they are independent when appart and become perfectly correlated when meeting each other. In Figure \ref{fig1}(c), $S^1_t, S^2_t,\dots, S^d_t$ are represented in different colours. 

Define $\tau^t_{ij} = \inf \{s \ge 0, \ S^i_t(t-s) =  S^j_t(t-s) \}$. By looking carefully at Figures \ref{fig1}(b) and \ref{fig1}(c), we let the reader convince herself that two individuals $i$ and $j$ have the same type at time $t$ iff: 
\begin{itemize}
\item[(i)] $\tau^t_{ij} \le t$ and  there are no $\star$ in the paths of $S^i_t$ and $S^j_t$ before $\tau^t_{ij}$, or

\item[(ii)] $\tau^t_{ij} \ge t$ and $S^i(0) \sim_{\pi_0} S^j(0)$ and  $S^i_t$ and $S^j_t$ their is no $\star$ in their paths, where $\pi_0$ is the initial genetic partition of the metapopulation, that is random partition of law $P^0$. 
\end{itemize}

From here, it is easy to check that:
\begin{eqnarray*}
D_t(i,j) & =& 1 - P_t(\{ \pi, i\sim_{\pi} j) \nonumber \\
& = &  \ 1 \ - \  \int_0^t e^{-2b_\infty s} \Pp(\tau_{ij} \in ds) ds \ - \ \int_{\pi} \ e^{-2b_\infty t} \Pp\left(\tau_{ij} > t, S^i(t) \sim_{\pi} S^j(t)\right) \ P^0(d\pi) . \nonumber 
\end{eqnarray*}

As $\forall i,j \in E$, $(D_t(i,j))$ is continuous, the fact that $(d_t^\eps(i,j))$ converges in distribution (in the weak topology) to $ (D_t(i,j))$  \eqref{convergence:eps} implies (by the continuous mapping theorem) that $(d_t^\eps(i,j))$ converges to $(D_t(i,j))$ in the sense of finite dimensional distributions, as $\eps \to 0$.

This result, combined with Theorem \ref{IBMtoPBM}, also implies that: $$\lim_{\eps \to 0} \lim_{\gamma \to 0 } (d^{\gamma, \eps}_{t/(\gamma \eps)}(i,j), \ t \ge 0) \ = \ D_t(i,j), \ t\ge 0)\  \textrm{ in the sense of finite dimensional distributions.}$$

The fact that $\lim_{t \to \infty} D_t(i,j) \ = \ 1 - \E(e^{-2 b_\infty \tau_{ij}})$  is a direct consequence of the definition of $(D_t(i,j); t\ge0)$ and the dominated convergence theorem.

This completes the proof of Theorem \ref{thm-intro}. \end{proof}

\section{An example: a population with a geographic bottleneck}
\label{sec:bottleneck}
Fix $d\in \N\setminus\{0\}$. We let ${\cal G}_1$ and ${\cal G}_2$ be two complete graphs of  $d$ vertices. We link the two graphs ${\cal G}_1$ and ${\cal G}_2$ by adding an extra edge $(v_1,v_2)$, where $v_k, k=1,2$ is a given vertex in ${\cal G}_k$. We call ${\cal G}$ the resulting graph. We equip ${\cal G}$ with the following migration rates: if $i$ is connected to $j$, then $M_{ij}=1/d$ (so that the emigration rate from any vertex $i$ is $1$ if $i\neq v_1,v_2$ and $1+\frac{1}{d}$ otherwise). We also assume that $N_i=1$, so that $\tilde M_{ij}=1/d$.

We think of ${\cal G}$ as  two well-mixed populations connected by a single geographic bottleneck. 

\begin{theorem}\label{thm:ex}
Fix $c>0$, $b_\infty = \frac{c}{d}$. Then for any two neighbours $i,j\in {\cal G}$
\begin{eqnarray*}
1-\E\left(\exp(-2b_\infty \tau_{ij})\right) \ = \left\{ \begin{array}{cc} \ \frac{c}{1+c} + o(1) & \mbox{if $i,j\in {\cal G}_1$, \ or if $i,j\in {\cal G}_2$} 
\\ 1-\frac{1}{d} + o(\frac{1}{d}) & \mbox{if $i=v_1$ and $j=v_2$.}
  \end{array} \right.
\end{eqnarray*}
\end{theorem}

Before going into the details of the proof, we present some heuristics for the formulae. 
First, if $i$ and $j$ belong to the same subgraph, we can assume that most of the time the two random walks hit each other before hitting the other subgraph. So we can consider  $\tilde S^i$ and $\tilde S^j$, two random walks on a complete graph with $d$ vertices. If $J$ is the number of jumps made by $\tilde S^i$ and $\tilde S^j$ before hitting each other, $J$ follows a geometric distribution of parameter $1/d$, which properly renormalized converges to  an exponential distribution of parameter $1$ (when $d\gg1$). In addition, the mean time between two consecutive jumps (of $S^i$ or $S^j$) is $2 \times 1/d$ so the distribution of $\tau_{ij}$ can be approximated by an exponential distribution of parameter $2/d$. If $e$ is an exponentially distributed random variable, with parameter $2/d$, then  $\E\left(\exp(-2b_\infty e)\right) = 1/(1+ 2b_\infty/d) = 1/(1+c)$ which gives the desired result. 
Second, $i=v_1$ and $j=v_2$, with  probability $1/(d+1) \simeq 1/d$, the first jump is from $v_1$ to $v_2$ (or $v_2$ to $v_1$), so the two random walks hit very fast, and the genetic distance is close to $0$. Otherwise, each random walk ``gets lost'' in its subgraph and the hitting time becomes very large. In that case, the genetic distance is approximatively $1$. 

\bigskip

\begin{proof} We give a brief sketch of the computations since the method is rather standard. We start with some general considerations. Consider a general meta-population with $\bar d$ subpopulations. 
Define $a(i,j) \ = \ \E\left(\exp(-2b_\infty \tau_{ij})\right)$. By conditioning on every possible move of the two walks on the small time interval $[0,dt]$, it is not hard to show that the $a(i,j)$'s satisfy the following system of linear equations: $\forall i\in\{1,\dots,\bar d\}, \ a(i,i)=1$ and
$\forall i,j\in\{1,\dots,\bar d\}$ with $i\neq j$:
\be\label{eq:general-tauij}
0 \ = \ \sum_{k=1}^{\bar d} \left( a(k,j) \tilde M_{ik} + a(i,k) \tilde M_{jk} \right) \ - \ a(i,j)\left(\sum_{k=1}^{\bar d} (\tilde M_{ik}+\tilde M_{kj}) + 2b_\infty \right). 
\ee
Let us now go back to our specific case (in particular $\bar d=2d$).
We distinguish between two types of points: the boundary points (either $v_1$ or $v_2$), and the interior points of the subgraphs ${\cal G}_1$ and ${\cal G}_2$ (points that are distinct from $v_1$ and $v_2$).
For $(i,j)$, with $i\neq j$, we say that $(i,j)$ is of type
\begin{itemize}
\item $(II)$ if the vertices belong to the interior of the same subgraph (either ${\cal G}_1$ or ${\cal G}_2$).
\item $(I\bar I)$ if the vertices belong to the interior of distinct subgraphs.
\item $(IB)$  if one of the vertex is in the interior of a subgraph, and the other vertex belongs to the boundary point
of the same subgraph.
\item $(I \bar B)$, $(B \bar B)$ are defined analogously.
\end{itemize}
By symmetry, $a(i,j)$ is invariant in each of those classes of pairs of points. We denote by $a(II)$ the value of $a(i,j)$ for $(i,j)$ in $(II)$.
$a(I \bar I), a(IB), a(I\bar B), a(B\bar B)$ are defined analogously. From this observation, we can inject those quantities in (\ref{eq:general-tauij}): this reduces the dimension of the linear problem from $\bar d(\bar d-1)$
to only $5$. The system can then be solved explicitly and straightforward asymptotics  yield Theorem \ref{thm:ex}.

\end{proof}

\section*{Acknowledgements}
We would like to thank the Editor and the reviewers for their useful comments on the previous version of the manuscript.
We also thank Florence D\'ebarre and  Amaury Lambert for helpful discussions.

\bibliographystyle{elsarticle-harv} 
    \bibliography{BiblioGeneral.bib}

\begin{thebibliography}{28}
\expandafter\ifx\csname natexlab\endcsname\relax\def\natexlab#1{#1}\fi
\expandafter\ifx\csname url\endcsname\relax
  \def\url#1{\texttt{#1}}\fi
\expandafter\ifx\csname urlprefix\endcsname\relax\def\urlprefix{URL }\fi

\bibitem[{Aldous(1989)}]{aldous1989}
Aldous, D., 04 1989. Stopping times and tightness. ii. Ann. Probab. 17~(2),
  586--595.

\bibitem[{Doyle and Snell(1984)}]{doyle}
Doyle, P., Snell, J., 1984. Random Walks and Electric Networks, 1st Edition.
  Vol.~22. Mathematical Association of America.

\bibitem[{Dubuisson and Jain(1994)}]{dubuisson1994modified}
Dubuisson, M.-P., Jain, A.~K., 1994. A modified hausdorff distance for object
  matching. In: Pattern Recognition, 1994. Vol. 1-Conference A: Computer Vision
  \& Image Processing., Proceedings of the 12th IAPR International Conference
  on. Vol.~1. IEEE, pp. 566--568.

\bibitem[{Fournier and M\'el\'eard(2004)}]{fournier2004}
Fournier, N., M\'el\'eard, S., 11 2004. A microscopic probabilistic description
  of a locally regulated population and macroscopic approximations. Ann. Appl.
  Probab. 14~(4), 1880--1919.

\bibitem[{Gavrilets(1997)}]{gavri-holey}
Gavrilets, S., 1997. Evolution and speciation on holey adaptive landscapes.
  Trends Ecol Evol. 12~(8), 307--312.

\bibitem[{Gavrilets et~al.(2000{\natexlab{a}})Gavrilets, Acton, and
  Gravner}]{Ga2000a}
Gavrilets, S., Acton, R., Gravner, J., 2000{\natexlab{a}}. Dynamics of
  speciation and diversification in a metapopulation. Evolution 54, 1493--1501.

\bibitem[{Gavrilets and Gravner(1997)}]{gavri1997}
Gavrilets, S., Gravner, J., 1997. Percolation on the fitness hypercube and the
  evolution of reproductive isolation. J Theor Biol. 184~(1), 51--64.

\bibitem[{Gavrilets et~al.(1998)Gavrilets, Li, and Vose}]{gavri1998}
Gavrilets, S., Li, H., Vose, M., 1998. Rapid parapatric speciation on holey
  adaptive landscapes. Proc. R. Soc. Lond. B 265, 1483--1489.

\bibitem[{Gavrilets et~al.(2000{\natexlab{b}})Gavrilets, Li, and
  Vose}]{gavri2000}
Gavrilets, S., Li, H., Vose, M., 2000{\natexlab{b}}. Patterns of parapatric
  speciation. Evolution 54~(4), 1126--1134.

\bibitem[{Griffiths(1981)}]{ARG1}
Griffiths, R.~C., 1981. Neutral two-locus multiple allele models with
  recombination. Theor. Popul. Biol. 19~(2), 169--186.

\bibitem[{Griffiths(1991)}]{ARG2}
Griffiths, R.~C., 1991. The two-locus ancestral graph. In: Basawa, I., Taylor,
  R.~L. (Eds.), Selected Proceeedings of the Symposium on Applied Probability.
  Institute of Mathematical Statistics, pp. 100--117.

\bibitem[{Hashimoto et~al.(2015)Hashimoto, Sun, and Jaakkola}]{Hashimoto}
Hashimoto, T.~B., Sun, Y., Jaakkola, T.~S., 2015. From random walks to
  distances on unweighted graphs. In: Proceedings of the 28th International
  Conference on Neural Information Processing Systems. NIPS'15. MIT Press,
  Cambridge, MA, USA, pp. 3429--3437.

\bibitem[{Hudson(1983)}]{hudson}
Hudson, R., 1983. Properties of the neutral model with intragenic
  recombination. Theor.Pop.Biol. 23~(2), 213--201.

\bibitem[{Ivanciuc(2000)}]{ivanciuc}
Ivanciuc, O., 2000. Qsar and qspr molecular descriptors computed from the
  resistance distance and electrical conductance matrices. ACH Models in
  Chemistry 5/6~(137), 607--632.

\bibitem[{Karlin(1982)}]{Karlin}
Karlin, S., 1982. (classification of selection-migration structures and
  conditions for a protected polymorphism. Evol. Biol., 61--204.

\bibitem[{Kimura(1968)}]{Kimura2}
Kimura, M., 1968. Evolutionary rate at the molecular level. Nature 217.

\bibitem[{Kimura and Ohta(1969)}]{Kimura1}
Kimura, M., Ohta, T., 1969. The average number of generations until extinction
  of an individual mutant gene in finite population. Genetics 3~(63).

\bibitem[{Klein and Randic(1993)}]{klein}
Klein, D., Randic, M., 1993. Resistance distance. Journal of Mathematical
  Chemistry~(12), 81 -- 95.

\bibitem[{Liben-Nowell and Kleinberg(2003)}]{liben}
Liben-Nowell, D., Kleinberg, J., 2003. The link prediction problem for social
  networks. International Conference on Information and Knowledge Management
  (CIKM), 556--559.

\bibitem[{Mal{\'e}cot(1948)}]{malecot}
Mal{\'e}cot, G., 1948. Les math{\'e}matiques de l'h{\'e}r{\'e}dit{\'e}.
  Barn{\'e}oud fr{\`e}res.

\bibitem[{Noest(1997)}]{noest}
Noest, A., 1997. Instability of the sexual continuum. Proc. R. Soc. Lond. B
  264, 1389--1393.

\bibitem[{Qiu and Hancock(2005)}]{qiu}
Qiu, H., Hancock, E., 2005. Image segmentation using commute times. Proceedings
  of the 16th British Machine Vision Conference (BMVC), 929--938.

\bibitem[{Roy(2004)}]{roy}
Roy, K., 2004. Topological descriptors in drug design and modeling studies.
  Molecular Diversity 4~(8), 321--323.

\bibitem[{Von~Luxburg et~al.(2014)Von~Luxburg, Radl, and Hein}]{von2014hitting}
Von~Luxburg, U., Radl, A., Hein, M., 2014. Hitting and commute times in large
  random neighborhood graphs. Journal of Machine Learning Research 15~(1),
  1751--1798.

\bibitem[{Wright(1932)}]{wright1932}
Wright, S., 1932. The roles of mutation, inbreeding, crossbreeding, and
  selection in evolution. Proceedings of the Sixth International Congress on
  Genetics, 355--366.

\bibitem[{Yamaguchi and Iwasa(2013)}]{Yama}
Yamaguchi, R., Iwasa, Y., 2013. First passage time to allopatric speciation.
  Interface Focus 3~(6).

\bibitem[{Yamaguchi and Iwasa(2015)}]{Yama2}
Yamaguchi, R., Iwasa, Y., 2015. Smallness of the number of incompatibility loci
  can facilitate parapatric speciation. Journal of Theoretical Biology 405,
  36--45.

\bibitem[{Yen et~al.(2005)Yen, Vanvyve, Wouters, Fouss, Verleysen, and
  Saerens}]{yen}
Yen, L., Vanvyve, D., Wouters, F., Fouss, F., Verleysen, M., Saerens, M., 2005.
  Clustering using a random walk based distance measure. In Proceedings of the
  13th Symposium on Artificial Neural Networks (ESANN), 317--324.

\end{thebibliography}

\end{document}